\documentclass[11pt]{iopart}
\usepackage{iopams}  

\usepackage{comment}
\usepackage{amssymb}
\usepackage{amsthm}
	\newtheorem{prop}{Proposition}[section]
		
\theoremstyle{definition}
	\newtheorem{definition}{Definition}[section]

\usepackage{upgreek}
\usepackage{nicefrac}
\usepackage{braket}

	\newcommand*{\vb}[1]{\boldsymbol{#1}} 
	\newcommand*{\dd}{\,\mathrm{d}} 

\newcommand*{\sub}[1]{_{\mathrm{#1}}}
\newcommand*{\ie}{i.e.\ }
\newcommand*{\eg}{e.g.\ }

\newcommand*{\R}{\mathbb{R}}
\newcommand*{\C}{\mathbb{C}}
\newcommand*{\N}{\mathbb{N}}
	
\newcommand*{\Z}{\mathbb{Z}}

\newcommand*{\iu}{\mathrm{i}\mkern1mu} 
\newcommand*{\pg}{\pi} 

\newcommand*{\Leb}{L} 

\newcommand*{\dom}{\mathfrak{D}} 

\newcommand*{\Bb}{\vb{B}} 						
\newcommand*{\Eb}{\vb{\mathcal{E}}} 			
\newcommand*{\Ef}{\mathcal{E}} 					

\newcommand*{\A}{\vb{A}}						
\newcommand*{\p}{\vb{p}}						

\newcommand*{\Os}{\Omega\sub{s}}					
\newcommand*{\dOs}{\partial\Omega\sub{s}}		

\newcommand*{\kad}{\mathsf{k}}
\newcommand*{\Bad}{\mathsf{B}}
\newcommand*{\Ead}{\mathsf{E}}
				
\newcommand*{\yy}{\zeta}

\newcommand*{\grad}{\vb{\nabla}}	
\newcommand*{\gradA}{\vb{\nabla}_{\!\A}}	

\newcommand*{\Fx}{\mathcal{F}_x} 

\newcommand*{\Cinf}{ \mathrm{C}_0^{\infty} }

\newcommand*{\AbsC}{\mathrm{AC}}
\newcommand*{\Schw}{\mathcal{S}} 

\newcommand*{\Herm}{\mathrm{H}}	

\newcommand*{\UU}{\mathrm{U}}	

\newcommand*{\htheta}{\uptheta}	
\newcommand*{\ddelta}{\updelta}	

\usepackage{xcolor}
	\definecolor{Verde}{cmyk}{1,0.21,1,0.2}
	\definecolor{Blu}{cmyk}{1,0.6,0,0.2}
	\definecolor{Rosso}{cmyk}{0.3,1,1,0.2}	
	\definecolor{Arancione}{cmyk}{0,0.8,1,0}	
	\definecolor{darkred}{cmyk}{0,1,1,0.3}	
	\definecolor{darkdarkred}{cmyk}{0,1,1,0.6}

\usepackage{siunitx} 
\usepackage{tikz}

\usepackage{pgfplots}
	\pgfplotsset{/pgf/number  format/1000 sep={\,}, compat=newest}
\usepgfplotslibrary{colormaps, fillbetween}

\pgfplotsset{
	basicoptions/.style={
				tick label style={font=\small},
				label style={font=\small},
				legend style={font=\small, fill=none},
				legend style={draw=none}, legend cell align=left,},	
}

\usepackage{hyperref}
\hypersetup{colorlinks=true, citecolor=Blu, linkcolor=Rosso, bookmarks}

\begin{document}

\title{Boundary conditions for the quantum Hall effect}

\author{Giuliano Angelone$^{1,2,*}$, Manuel Asorey$^{3}$, Paolo Facchi$^{1,2}$, Davide Lonigro$^{1,2}$, Yisely Martinez$^{3}$ }

\address{$^1$ Dipartimento di Fisica, Universit\`a di Bari, I-70126 Bari, Italy}
\address{$^2$ INFN, Sezione di Bari, I-70126 Bari, Italy}
\address{$^3$ Centro de Astropart\'{\i}culas y F\'{\i}sica de Altas Energ\'{\i}as, Departamento de F\'{\i}sica Te\'orica,
Universidad de Zaragoza, E-50009 Zaragoza, Spain}
\ead{$^*$giuliano.angelone@ba.infn.it}

\begin{abstract}
We formulate a self-consistent model of the integer quantum Hall effect on an infinite strip, using boundary conditions to investigate the influence of finite-size effects on the Hall conductivity. By exploiting the translation symmetry along the strip, we determine both the general spectral properties of the system for a large class of boundary conditions respecting such symmetry, and the full spectrum for (fibered) Robin boundary conditions. In particular, we find that the latter introduce a new kind of states with no classical analogues, and add a finer structure to the quantization pattern of the Hall conductivity. Moreover, our model also predicts the breakdown of the quantum Hall effect at high values of the applied electric field.
\end{abstract}


\maketitle

\section{Introduction}
Since its discovery, dating back to 1980~\cite{KliDoPe80}, the quantum Hall effect (QHE) has generated huge interest in the scientific community, both from the theoretical and experimental perspective. The quantization of the Hall conductivity, and in particular its surprising robustness (i.e.\ its independence from the details of the experimental setup), has stimulated many theoretical physicists to find a compelling explanation of the phenomenon.

A large number of approaches have been indeed considered, ranging from the phenomenological description of edge currents~\cite{Halp82, Butt88}, to the application of topology concepts in both lattice and continuous models~\cite{Kohmo85, Hatsu93, Hatsu93b, Hatsu97}, to random Hamiltonians~\cite{Wang97} and effective quantum field theories~\cite{ZhaHan89, CaCha91}. Besides, the QHE can be surely considered as the progenitor of topological matter, a subject nowadays very active that has recently expanded in many interesting directions~\cite{HaKa10, QiZha11}. An exhaustive survey of the literature is well beyond our scopes, and we refer the reader to the reviews~\cite{AvrOsa03, KliCha20}, as well as to the books~\cite{ChaPie95, Yo02} and to the lecture notes~\cite{Tong16}.

In this work we formulate a self-consistent model of the integer QHE, describing a bounded quantum Hall system by using suitable boundary conditions, instead of resorting to a confining potential, as it is usually done. Boundary conditions have already been applied in the past for the description of the QHE, with a focus on finite-size effects and on the phenomenology of the edge states~\cite{NiuTho87, JJV95, AANS98, DeBP99, CoHiSo02, MU17}, but their influence on the quantization of the Hall conductivity is still scarcely investigated. For our purposes boundary conditions provide a simple effective framework to describe the QHE, reducing the latter to its key ingredients: a charged particle, a bounded two-dimensional system, and a strong magnetic field. Despite its simplicity, our model will prove remarkably powerful, as we will predict both the quantization of the Hall conductivity as well as its breakdown at high values of the applied electric field. Moreover, as we will extensively discuss, novel phenomena appear in the quantum regime when certain boundary conditions are applied.

More generally, boundary conditions represent a versatile tool which allows one to model the interaction between the bulk of a system and its boundary, also providing a clear physical intuition of what is going on. In recent years, quantum boundary conditions have accordingly attracted an increasing interest in different branches of quantum physics, such as for the description of the Casimir effect~\cite{AsoAlvMun06, AsoAlvMun07, AsMu13, AsBaPe16}, topology change~\cite{BaBiMa95, IboPere15}, geometric phases~\cite{FaGaMa16}, topological insulators and QCD~\cite{AsBaPe13}, isospectrality and inverse spectral problems~\cite{isob, LaKu21}, and other phenomena such as spontaneous symmetry breaking, quantum anomalies~\cite{AsoMun12}, and entanglement generation~\cite{IbMaPe14}.

The present paper is organized as follows. We start with an extensive theoretical study of the model: in Section~\ref{sec-model} we introduce the Hall Hamiltonian with fibered, but otherwise arbitrary, boundary conditions. In Section~\ref{sec-generalspectrum} we present some general results regarding the general structure of the spectrum  and its asymptotic properties. In Section~\ref{sec-transport} we then move to  transport properties, defining the velocity operator and deriving a formula to compute the Hall conductivity of the system.

Armed with these instruments, in Section~\ref{sec-saext} we complete the discussion of the model. We first identify a  family of self-adjoint extensions, related to the Robin boundary conditions, which turns out to be relevant for the description of the QHE. We then determine the spectrum of the system and its dependence on the external fields and on boundary conditions, propose a semi-classical picture, and investigate the finite-size effects on the quantization of the Hall conductivity.

\section{Quantum Hall Hamiltonian with fibered boundary conditions}\label{sec-model}
By a  closed \emph{quantum Hall syste} we denote  an ensemble of  electrons confined in an effective two-dimensional substrate, the \emph{Hall device}, that is kept at a sufficiently low temperature and is subjected to a magnetic field, perpendicular to the device, and to a (possibly vanishing) electric field perpendicular to the magnetic one. For our purposes, we will assume the electrons to be  non-relativistic and non-interacting. As it turns out, interactions seems to be necessary in order to describe the fractional regime of the QHE, but can be neglected in the integer regime, to which we are interested~\cite{ChaPie95, Yo02}. On the other hand, we mention that relativistic effects do actually play a role in some quantum Hall systems such as, for example, when the substrate is made of graphene~\cite{JZ07, No07}. Moreover, for the sake of simplicity, we will model electrons as spinless particles, although spin can be added to our model without substantial effort. 

By neglecting the mutual interaction between the electrons, the properties of the total Hamiltonian can  be easily derived from those of the single-particle Hamiltonian $H$, which can be written as $H=H\sub{fields}+V\sub{device}$, with $H\sub{fields}$ and $V\sub{device}$ describing the interaction of the electron with the external fields and with the Hall device, respectively. In particular, the latter term contains information  on both the microscopical (lattice) structure of the Hall device and its macroscopic shape. In the following we are going to consider a \emph{coarse-grained} and \emph{hard-wall} model, by assuming \ie that $V\sub{device}$ vanishes inside the Hall device and is (formally) infinite on its exterior. In this way, electrons are effectively constrained in a certain region $\Omega$ of the plane, which macroscopically corresponds to the geometry of the Hall device, and the microscopic details of $V\sub{device}$ are encoded into suitable boundary conditions. Further details about some possible spectral effects of $V\sub{device}$ are given in~\ref{sec:floquet}.

\begin{figure}[tb]
\centering
\includegraphics{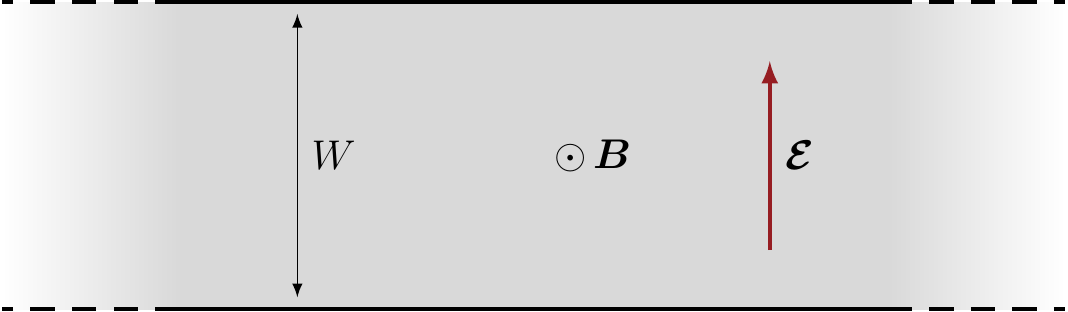}
\caption{Sketch of a Hall strip, consisting in an infinite strip of width $W$ subjected to the external crossed fields $\Bb=B\hat{\vb{z}}$ and $\Eb=\Ef\hat{\vb{y}}$.}
\label{fig-strip}
\end{figure}

\subsection{The Hall Hamiltonian}
Consider a non-relativistic spinless particle with mass $m$ and negative electric charge $q=-e$, constrained in an infinite strip of width $W>0$, 
\begin{equation}\label{eq-Omegas}
\Os \equiv \{(x,y)\in\R\times I_W\}\subset\R^2\,,\qquad I_W \equiv ({-}W/2,W/2)\,,
\end{equation}
and subjected to the external  fields
\begin{equation}
\Bb=(0,0,B)=B\hat{\vb{z}}\qquad\textnormal{and}\qquad \Eb=(0,\Ef,0)=\Ef\hat{\vb{y}}\,,\label{eq-BEfields}
\end{equation}
\ie to a uniform magnetic field orthogonal to the surface of the strip and to a uniform transversal electric field. This setup, which we will hereafter denote as a \emph{Hall strip}, is depicted in Figure~\ref{fig-strip}. The Hamiltonian of the system, which will be referred to as the \emph{Hall Hamiltonian}, takes thus the form
\begin{equation}\label{eq-HAE}
H_{\A}\equiv -\frac{\hbar^2}{2m}\gradA^{2}+e\Ef y\,,
\end{equation}
where $\hbar\equiv h/(2\pg)$ is the reduced Planck constant,
\begin{equation}\label{eq-gradA}
\gradA\equiv \grad+\iu\frac{e}{\hbar}\A=\Biggl(\frac{\partial}{\partial x}+\iu\frac{e}{\hbar}A_x, \frac{\partial}{\partial y}+\iu\frac{e}{\hbar}A_y\Biggr)
\end{equation}
denotes the covariant derivative, and $\A=(A_x,A_y)$ is a (suitably regular) vector field satisfying  
\begin{equation}\label{eq-curlA}
B=\frac{\partial A_y}{\partial x}-
\frac{\partial A_x}{\partial y} \,,
\end{equation}
and represents the vector potential associated with the magnetic field $\Bb=B\hat{\vb{z}}$. Since the electron is constrained in the strip $\Os$, the Hamiltonian $H_{\A}$ acts on the Hilbert space   $\Leb^2(\Os)$
of complex square-integrable functions on $\Os$, endowed with the 
scalar product
\begin{equation}
 \braket{\psi|\phi}
=\int_{\Os} \psi^{*}(x,y)\,\phi(x,y) \dd{x}\dd{y}
\end{equation}
and its associated norm, $\| \psi\|^2= \braket{\psi|\psi}$.

For later use, we note that, since $\Os=\R\times I_W$ is a Cartesian product, $\Leb^2(\Os)$ naturally decomposes into the tensor product
\begin{equation}
\Leb^2(\Os)\cong \Leb^2(\R)\otimes \Leb^2(I_W)\,.
\end{equation}

The vector potential $\A$ is not completely fixed by the relation~\eref{eq-curlA}: as a matter of fact, this introduces a gauge freedom, \ie an ambiguity in the choice of the Hamiltonian. Despite the fact that many physically observable quantities, such as the spectrum of $H_{\A}$, are independent of the gauge choice~\cite{Lein83}, such a choice is needed in order to perform calculations. More importantly, one should be aware that, generally, even boundary conditions depend on the gauge choice: see~\cite{quanBill} for a detailed discussion. For the Hall strip \eref{eq-Omegas}--\eref{eq-BEfields}, a profitable choice is the Landau gauge
\begin{equation}\label{eq-Landau}
\A=(-By, 0)\,,
\end{equation}
which clearly satisfies Eq.~\eref{eq-curlA}. In this gauge, the Hall Hamiltonian reads
\begin{equation}\label{eq-HBE}
H=H(B,\Ef)=-\frac{\hbar^2}{2m}\Biggl(\frac{\partial}{\partial x} - \iu \frac{e}{\hbar}By\Biggr)^2 -\frac{\hbar^2}{2m} \frac{\partial^2}{\partial y^2} +e\Ef y\,,
\end{equation}
and it is invariant under longitudinal translations, \ie we formally have that
\begin{equation}\label{eq-xInv}
[ p_x, H]=0\,,
\end{equation} 
where $p_x$ is the first component of the  momentum operator $\p = -\iu\hbar\grad$. A precise meaning of the above equation is given in~\ref{sec-symmetry}.

Notice that Eq.~\eref{eq-HBE} (as well as Eq.~\eref{eq-HAE}) does not suffice to define an operator on $\Leb^2(\Os)$, since the formal expression $H\psi$ does not give a square-integrable function for each $\psi\in \Leb^2(\Os)$. Accordingly, the specification of a domain is needed in order to have a well-defined unbounded operator on $\Leb^2(\Os)$. Besides, in order to correspond to a physical observable and to generate a unitary evolution, the Hall Hamiltonian must be a (essentially) \emph{self-adjoint} operator on the chosen domain. For differential operators, as in our case, a domain specification substantially corresponds to a proper choice of \textit{boundary conditions}. Physically, no legitimate observable can indeed be obtained without specifying the behavior of wavefunctions at the boundary. The following strategy is usually applied when dealing with such operators:
\begin{itemize}
	\item first, a suitable domain $\dom_0$ of smooth functions that vanish in a neighborhood of  the boundary is chosen for the operator, thus obtaining a symmetric, but not self-adjoint operator, which is referred to as the \textit{minimal realization};
	\item then, one computes its adjoint, which is the \textit{maximal realization} of the operator, and imposes suitable boundary conditions on the domain of the latter.
\end{itemize}

In this paper we will follow a slightly different route that, while less general, will fully enable us to describe the physical situation that we have in mind. Our construction will crucially rely on the decomposition of $H$ into the direct integral of a certain family of operators (\emph{fibers}), each of them acting on the reduced Hilbert space $\Leb^2(I_W)$. This will allow us to reduce a two-dimensional problem to a continuous family of one-dimensional problems, which can be individually solved and then glued back together. As a matter of fact this procedure drastically simplifies both the search for a suitable domain which renders the full Hamiltonian $H$ self-adjoint and the computation of its spectrum. It will be indeed sufficient to analyze the self-adjointness and the spectrum of each fiber operator separately.

From a physical point of view, the direct integral decomposition is essentially a consequence of  translational invariance of the system we are studying. This symmetry is in turn ensured both by our gauge choice and by a suitable choice of  boundary conditions: not all the self-adjoint extensions of $H$ do indeed satisfy the invariance encoded in Eq.~\eref{eq-xInv}, as its right-hand side depends non-trivially on the actual domain of $H$. In other words, boundary conditions can eventually break the translational symmetry of the system.\footnote{We explicitly prove this statement in~\ref{sec-symmetry}. Instead in the main text we adopt a bottom-up (\ie constructive) approach.} Accordingly, in the following we will only consider boundary conditions (and thus self-adjoint extensions) which preserve the translational symmetry, being thus compatible with the direct integral decomposition: we shall refer to them as \emph{fibered} boundary conditions.

\subsection{Minimal realization and its fibers}\label{subsec-minim}
We  start by defining $H$ on the minimal domain
\begin{equation}\label{eq-dom0}
\dom_0\equiv\Schw(\R)\otimes\Cinf(I_W)\subset \Leb^2(\R)\otimes\Leb^2(I_W)\,.
\end{equation}
Here $\Cinf(I_W)$ is the space of smooth functions compactly supported on the segment $I_W$, and thus vanishing in a neighborhood of the boundary $\dOs$,
whereas $\Schw(\R)$ denotes the Schwartz space of rapidly decreasing smooth functions. 

As it turns out, $H$ is symmetric but not self-adjoint when defined on $\dom_0$ (see~e.g.~\cite{AIM15}).

Recall that any function $\psi\in\Schw(\R)$ admits a Fourier transform,
\begin{equation}
(\Fx\psi)(k)=\frac{1}{\sqrt{2\pg}}\int_{\R} \e^{-\iu kx}\psi(x)\dd{x}\,,
\end{equation}
where $k$, the Fourier conjugate of $x$, represents the \emph{wavenumber}. Physically, it is equal to $1/\hbar$ times the momentum of the particle $p_x$. Moreover:
\begin{itemize}
	\item $\Fx\colon\Schw(\R)\to\Schw(\R)$ is bijective and isometric;
	\item $\Fx$ can be extended to an unitary operator on the whole space $\Leb^2(\R)$.
\end{itemize}
In the following, we will write $\Fx\Schw(\R)=\Schw(\hat{\R})$ when we need to emphasize the distinction between the position representation ($\R$) and the momentum  one ($\hat\R$). Besides, we implicitly extend $\Fx$ to ${\Fx}\otimes{1}$ when acting on (subsets of) the full Hilbert space $\Leb^2(\R)\otimes \Leb^2(I_W)$. In this case, it corresponds to a partial Fourier transform. 

To take advantage from  translational invariance, it is convenient to switch from the original $(x, y)$-representation to the mixed $(k,y)$-representation: 
\begin{equation}
\Schw(\R)\otimes\Cinf(I_W)\quad\stackrel{\Fx}{\longrightarrow}\quad\Schw(\hat\R)\otimes\Cinf(I_W)\,.
\end{equation}
In the $(k,y)$-representation, $H$ acts as the unitarily equivalent operator
\begin{equation}\label{eq-hatHBE}
\hat H\equiv\Fx H\Fx^{-1}\,,
\end{equation}
which is thus defined on
\begin{equation}\label{eq-hatHBEdom}
\Fx\dom_0=\Schw(\hat\R)\otimes\Cinf(I_W)\,.
\end{equation}
Moreover, as it is customary to do, we henceforth canonically identify the tensor product space $\Leb^2(\hat{\R})\otimes \Leb^2(I_W)$ with  $\Leb^2(\hat{\R}; L^2(I_W))$.  Recall that  $\Leb^2(\hat{\R}; L^2(I_W))$ is defined as the space of (equivalence classes of) functions $\psi\colon\R\to\Leb^2(I_W)$ which are strongly measurable and such that 
\begin{equation}
\int_{\R} \|\psi(k)\|_{\Leb^2(I_W)}^2\dd{k}<\infty\,,
\end{equation}
see~\cite{HNVW} for more details. At this point, it is easy to show that $\hat H$ acts fiber-wise as 
\begin{equation}\label{eq-actiontildeHBE}
	(\hat H \psi)(k)=h(k)\psi(k)
\end{equation} 
on every component $\psi(k)=\psi_k(\cdot)\in\Cinf(I_W)$ of $\psi\in  \Fx\dom_0\subset \Leb^2(\hat{\R}; L^2(I_W))$, each fiber operator $h(k)$ being defined on $\Cinf(I_W)$ and having the expression
\begin{equation}\label{eq-Hk}
	h(k)=h(k; B, \Ef)= -\frac{\hbar^2}{2m}\frac{\mathrm{d}^2}{\mathrm{d} y^2}+V_k(y)\,,
\end{equation}
where the potential $V_k(y)$ is given by
\begin{equation}\label{eq-Vk}
	V_k(y)\equiv\frac{1}{2}m\omega^2_B(y- kl_B^2)^2+e\Ef y
\end{equation}
and where the quantities $\omega_B\equiv eB/m$ and $l_B\equiv \sqrt{\hbar/eB}$ are respectively known as the \emph{cyclotron frequency} and the \emph{magnetic length}~\cite{Tong16}. The physical significance of the fiber operator $h(k)$ will be clarified in Subsection~\ref{sec-scaling}.

\subsection{Self-adjoint extensions of the fibers}
Again, each fiber $h(k)$ is symmetric but not self-adjoint on $\Cinf(I_W)$ (as it must be since it is defined on a domain of functions vanishing near the boundary). The Schr\"{o}dinger operator $h(k)$ represents a very particular case of a second-order differential strongly elliptic operator. Self-adjoint realizations of such operators can be completely characterized in terms of boundary conditions~\cite{Gru68, FaGaLi18, Gru12}. In general, such characterization is subjected to non-trivial technical intricacies, involving \eg the introduction of Sobolev spaces of negative order and the definition of suitable trace operators.

Luckily, however, the discussion is greatly simplified when taking into account differential operators on one-dimensional regions (and in particular when $\Omega$ is a finite interval), which is precisely the case of each fiber $h(k)$ of the Hall Hamiltonian. In this case, indeed, the boundary $\partial\Omega$ is just given by the points $\{-W/2, W/2\}$, and the \emph{boundary data} are then completely encoded in two $\C^2$ vectors
\begin{equation}\label{eq-boundarydata}
\Psi\equiv
\left(\begin{array}{@{}c@{}}
\psi(-W/2)\\ \psi(W/2)
\end{array}\right)
\qquad\textnormal{and}\qquad \Psi' \equiv l_0
\left(\begin{array}{@{}c@{}}
-\psi'(-W/2)\\ \psi'(W/2)
\end{array}\right)\,,
\end{equation}
where $l_0>0$ is an arbitrary reference length. Moreover, for all $k\in\mathbb{R}$, the adjoint $h^{\dagger}(k)$ of the operator $h(k)$ is defined on the larger (common) domain
\begin{equation}\label{eq-Ddagger}
\dom\bigl(h^{\dagger}(k)\bigr)= H^2(I_W)\,,
\end{equation}
\ie on the Sobolev space of second order on the segment $I_W$ (see \eg Proposition~2.3.20 of~\cite{deO08}). The following result can then be proved~\cite{AIM05, BFV01}; see also Theorem~7.2.9 of~\cite{deO08}.
\begin{prop}\label{prop-characteriz}
	For each $k\in\mathbb{R}$, all self-adjoint extensions of $h(k)$ are in one-to-one correspondence with the set of $2\times 2$ unitary matrices $U(k)$ as follows: the matrix $U(k)$ defines a self-adjoint extension of $h(k)$ via the restriction
\begin{equation}\label{eq-HUk}
		h_{U}(k)\equiv h^\dag(k)\big|_{\dom(h_{U}(k))}\,,
\end{equation}
where
\begin{equation}\label{eq-bc-kdep}
\dom(h_{U}(k))\equiv \bigl\{\psi\in H^2(I_W) : \iu(I+U(k))\Psi=(I-U(k))\Psi'\bigr\}
\end{equation}
and where $I$ denotes the $2\times 2$ identity matrix.
\end{prop}
In other words, each matrix $U(k)$ implements the boundary condition
\begin{equation}\label{eq-Uboundary}
\iu(I+U(k))\Psi=(I-U(k))\Psi'\,,
\end{equation}
which in turn prescribes a linear relation between the boundary data $\Psi$ and $\Psi'$. 
For example, $U(k)=I$ corresponds to Dirichlet boundary conditions $\Psi=0$, while $U(k)=-I$ describes Neumann boundary conditions $\Psi'=0$. See Eq.~\eref{eq-boundarydata}. In general, since any $2\times 2$ unitary matrix can be parametrized by four real angles, all admissible boundary conditions for $h(k)$ will be parametrized by certain $k$-depending angles, say $\{\theta_i(k)\}_{i=1,\dots,4}$.
We stress that, in general, every self-adjoint realization of $h(k)$ will depend on the parameter $k$ \emph{both} via its expression, given by Eq.~\eref{eq-Hk}, \emph{and} via its domain, through the boundary condition.

\subsection{Direct integral and fibered boundary conditions}
Up to now, we have shown that each fiber $h(k)$ can be made a self-adjoint operator by choosing a suitable boundary condition. This property, together with Eq.~\eref{eq-actiontildeHBE}, will enable us to associate a self-adjoint extension of $\hat H$ (and thus of $H$) with each suitably regular function
\begin{equation}
	k\in\mathbb{R}\mapsto U(k)\in\UU(2)\,.
\end{equation}
The notion of direct integral will be of primary importance here. We thus recall some basic notions, referring \eg to~\cite{RS4} for more details. 
\begin{definition}\label{def-directint}
Let $\mathcal{H}$ be a Hilbert space and $\{a(k)\}_{k\in\mathbb{R}}$ a collection of self-adjoint operators on $\mathcal{H}$, each with domain $\dom_k$, such that the function
\begin{equation}
	k\in\mathbb{R}\mapsto\left(a(k)+\iu\right)^{-1}\in\mathcal{B}(\mathcal{H})
\end{equation}
is weakly measurable. The \emph{direct integral}
\begin{equation}
A=\int_{\mathbb{R}}^\oplus a(k)\,\dd k
\end{equation}
is defined as an operator on  $L^2(\mathbb{R}; \mathcal{H})$ with domain
\begin{equation}
		\hskip-50pt 
		\dom(A)=\Biggl\{\psi\in \Leb^2(\R; \mathcal{H}):
		\psi(k)\in \dom_k \textnormal{ for a.e. } k\in\R\,,
		\int_{\R} \|a(k)\psi(k)\|^2_{\mathcal{H}}\dd{k}<\infty\Biggr\}\,
\end{equation}
and such that $(A\psi)(k)=a(k)\psi(k)$ for almost every $k\in\mathbb{R}$.
\end{definition}
In particular, the following properties hold (see Theorem~XIII.85 of~\cite{RS4}).
\begin{prop}\label{prop-dirint}
The direct integral $A$ of a family of self-adjoint operators $\{a(k)\}_{k\in\mathbb{R}}$ on $\mathcal{H}$ is a self-adjoint operator on $L^2(\mathbb{R}; \mathcal{H})$. Moreover, its spectrum is given by
\begin{equation}\label{eq-specdirectint}
\hskip-30pt	\sigma(A)=\bigl\{\lambda\in\mathbb{R}:\forall\epsilon>0,\,\mu\bigl(\{k\in\R:{\sigma(a(k))}\cap(\lambda-\epsilon,\lambda+\epsilon)\neq\emptyset\}\bigr)>0\bigr\}\,,
\end{equation}
where $\mu(\cdot)$ denotes the Lebesgue measure.
\end{prop}
At this point, we are finally ready to attack our problem.
\begin{prop}\label{prop-dirintbis}
	Let $\{U(k)\}_{k\in\mathbb{R}}$ be a family of $2\times 2$ unitary matrices such that the function
	\begin{equation}
		k\in\mathbb{R}\mapsto U(k)\in\UU(2)
	\end{equation}
	is continuous, and let $\{h_{U}(k)\}_{k\in\R}$ be the family of self-adjoint operators defined by Eqs.~\eref{eq-HUk}--\eref{eq-bc-kdep}. Then the direct integral
	\begin{equation}\label{eq-directint}
	\hat{H}_{U}\equiv \int_{\R}^{\oplus}h_{U}(k)\dd{k}
	\end{equation}
	is a self-adjoint extension of the operator $\hat H$ on $\Fx\dom_0\subset \Leb^2(\hat{\mathbb{R}};L^2(I_W))$.
\end{prop}
\begin{proof}
By Definition~\ref{def-directint}, the direct integral~\eref{eq-directint} is well-defined as long as the function
\begin{equation}
	k\in\mathbb{R}\mapsto(h_{U}(k)+\iu)^{-1}\in\mathcal{B}\left(\Leb^2(I_W)\right)
\end{equation}
is weakly measurable; as we prove in~\ref{sec-resolvent}, a sufficient condition for this is the continuity of the function $k\mapsto U(k)$, which holds by assumption. Moreover, by Proposition~\ref{prop-dirint}, $\hat{H}_{U}$ is also self-adjoint. To prove that the latter is actually an extension of $\hat H$, just notice that its domain
\begin{equation}
 \eqalign{
	\dom( \hat{H}_{U} )=\Biggl\{\psi\in \Leb^2(\hat{\R}; \Leb^2(I_W)):{}&\psi(k)\in \dom(h_{U}(k)) \textnormal{ for a.e. } k\in\R\,,\cr
	&\int_{\R} \|h_{U}(k)\psi(k)\|^2_{L^2(I_W)}\dd{k}<\infty\Biggr\}\,
}
\end{equation}
clearly contains the minimal domain $\Fx\dom_0$ of $\hat H$, so that the fiber-wise action expressed by Eq.~\eref{eq-actiontildeHBE} still holds (for almost every $k\in\R$) by replacing each operator with its corresponding extension.
\end{proof}

The continuity request for $k\mapsto U(k)$ means that, loosely speaking, we are requiring the boundary conditions to depend “sufficiently nicely” on the parameter $k$; this is indeed the case in all physically realistic situations.  \emph{A fortiori}, the claim also holds if $U(k)=U_0$ does not depend on $k$, \ie if we impose the same boundary condition, associated with the unitary matrix $U_0$, on all the fibers; such are the classes of boundary conditions that we will ultimately analyze in Section~\ref{sec-saext}, and which are the ones of interest for the description of the QHE.

In any case, once a choice for all the $U(k)$ has been made, 
we can in principle go back to the position representation by taking
\begin{equation}
\dom(H_{U})=\Fx^{-1}\dom(\hat{H}_{U})
\end{equation}
as the domain of the Hall Hamiltonian. This procedure may be in general non-trivial, see~\ref{sec-position} for a simple example. However, as long as we are interested in the spectral properties of $H_{U}$, we can take advantage of the fact that $H_{U}$ and $\hat{H}_{U}$, being unitarily equivalent, share the same spectrum.

\section{General spectral properties}\label{sec-generalspectrum}
We devote this section to understand some general properties of the spectrum $\sigma(H_{U})$; while the particular features of this spectrum depend on the choice of the boundary conditions, its general structure and some asymptotic properties are independent of this choice. Note that hereafter we will always require the function $k\mapsto U(k)$ to be continuous, so that Proposition~\ref{prop-dirintbis} applies.

\subsection{Band structure of the spectrum}
We now show that, under our assumptions, the spectrum of $H_{U}$ exhibits a band structure. A first important observation is the following: for all $k\in\mathbb{R}$, all self-adjoint extensions of $h(k)$ have a purely discrete spectrum, see \eg Theorem~10.6.1 of~\cite{Ze05}. We can thus write that\footnote{In order to lighten our notations, we drop the dependence of $E_n(k)$ (and of related quantities) on the boundary conditions, as the latter will always be clear from the context.}
\begin{equation}
	\sigma(h_{U}(k))=\{E_n(k)\}_{n\in\N}
\end{equation}
with $E_n(k)$ to be obtained by solving the $k$-dependent eigenvalue problem 
\begin{equation}\label{eq-eigeneq}
\bigl[h_{U}(k) -E_n (k)\bigr]\psi_{n}(k; y)=0\,,
\end{equation}
where $\psi_n(k; y)\in\dom(h_{U}(k))$ denotes the $n$-th eigenfunction of $h_{U}(k)$.
\begin{prop}
	Let the function $k\in\mathbb{R}\mapsto U(k)\in\UU(2)$ be continuous. Then the spectrum of $H_{U}$ is given by
	\begin{equation}
	\sigma(H_{U})=\bigcup_{n\in \N} \Delta_n\,,
	\end{equation}
	each $\Delta_n$ being defined as the union of the $n$-th eigenvalues of each fiber $h_{U}(k)$:
	\begin{equation}
	\Delta_n\equiv \{ E_n(k) : k\in \R \}\,.
	\end{equation}
\end{prop}
\begin{proof}
	Recall that the eigenvalues of an operator can be characterized as the isolated singularities of its resolvent. By the regularity result presented in~\ref{sec-resolvent} and the continuity of $k\mapsto U(k)$, the resolvent of $h_{U}(k)$ is clearly a continuous function of $k$ and thus so are its poles: this means that, for all $n$, the map $k\mapsto E_n(k)$ is everywhere continuous (except at most on sets of vanishing measure). Now, by Proposition~\ref{prop-dirint}, the spectrum of a direct integral can be characterized via Eq.~\eref{eq-specdirectint}, which in our case reads
	\begin{eqnarray*}
	\hskip-45pt \eqalign{ 
	\sigma(H_{U})&=\Bigl\{\lambda\in\mathbb{R}:\forall\epsilon>0,\,\mu\bigl(\{k\in\R:{\sigma(h_{U}(k))}\cap(\lambda-\epsilon,\lambda+\epsilon)\neq\emptyset\}\bigr)>0\Bigr\}\cr
		&=\bigcup_{n\in\mathbb{N}}\Bigl\{\lambda\in\mathbb{R}:\forall\epsilon>0,\,\mu\bigl(\{k\in \R:{\{E_n(k)\}}\cap(\lambda-\epsilon,\lambda+\epsilon)\neq\emptyset\}\bigr)>0\Bigr\}\cr
		&=\bigcup_{n\in\mathbb{N}}\Bigl\{\lambda\in\mathbb{R}:\forall\epsilon>0,\,\mu\bigl(\{k\in\R:|E_n(k)-\lambda|<\epsilon\}\bigr)>0\Bigr\}\cr
		&=\bigcup_{n\in\mathbb{N}}\mathrm{ess.im}\,\left(E_n(\cdot)\right)
		\cr
		&=\bigcup_{n\in\mathbb{N}}\Delta_n\,,
	}
	\end{eqnarray*}
	where $\mathrm{ess.im}\,\left(E_n(\cdot)\right)$ represents the (Lebesgue-)essential range of the function $k\mapsto E_n(k)$, and in the last step we used the continuity of the latter. 
\end{proof}

The above result shows that, even at the spectral level, self-adjoint extensions of the Hall Hamiltonian $H$ via fibered boundary conditions are particularly easy to work with: their spectral properties can be inferred by solving a continuous family of one-dimensional eigenvalue equations.

\subsection{Scaling properties and electric spectrum}\label{sec-scaling}

\begin{figure}[tb]
\centering
\includegraphics{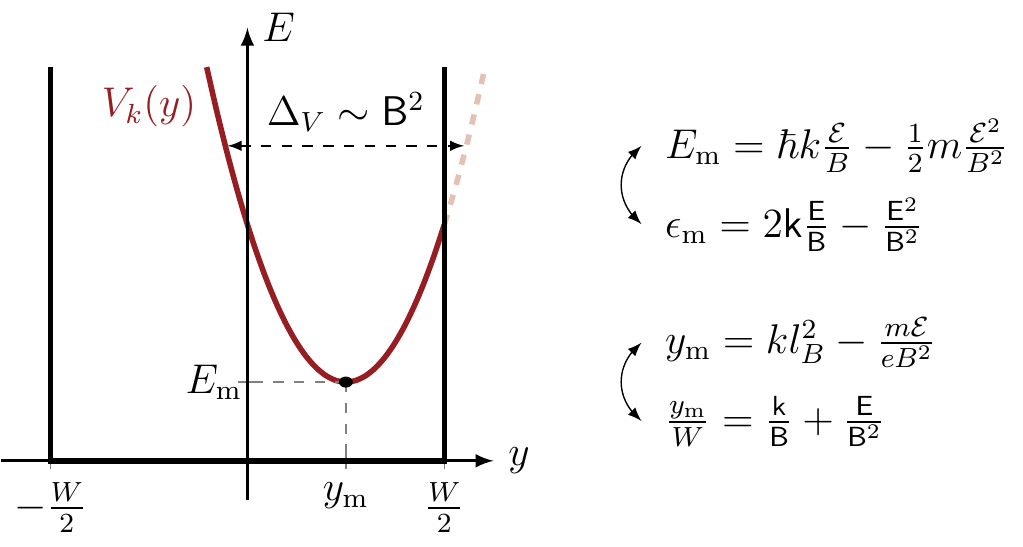}
\caption{Effective harmonic potential $V_k(y)$ inside the one-dimensional cavity $I_W$; its  (indicative) width $\Delta_V$ and the location $(y\sub{m}, E\sub{m})$ of its minimum are determined by $k$, $B$ and $\Ef$ or, equivalently, by the corresponding dimensionless quantities, see Eq.~\eref{eq-kBE}.}
\label{fig-costrainedHO}
\end{figure}

By our previous discussion, each fiber $h_{U}(k)$ takes the form of a one-dimensional Schr\"{o}dinger operator
\begin{eqnarray}
h_{U}(k)&=
-\frac{\hbar^2}{2m}\frac{\mathrm{d}^2}{\mathrm{d} y^2}
+\frac{1}{2}m\omega^2_B(y- kl_B^2)^2+e\Ef y\\
&=-\frac{\hbar^2}{2m}\frac{\mathrm{d}^2}{\mathrm{d} y^2}+\frac{1}{2}m\omega^2_B\Biggl(y-kl_B^2+\frac{m\Ef}{eB^2}\Biggr)^2+\hbar k\frac{\Ef}{B}-\frac{1}{2}m\frac{\Ef^2}{B^2}\,.\label{eq-Hk2}
\end{eqnarray}
Therefore, $h_{U}(k)$ represents the Hamiltonian of a quantum harmonic oscillator constrained in a one-dimensional cavity with some boundary conditions, the external fields $B$ and $\Ef$ controlling the width of the harmonic potential and the location of its minimum, see Figure~\ref{fig-costrainedHO}. Also notice that the wave-number $k$ influences the location of the potential minimum: accordingly, the spectrum of $h_{U}(k)$ depends on both the boundary conditions and on the three parameters $k$, $B$ and $\Ef$. 

Hereafter, when we need to explicit the dependence on the external fields, rather than on the boundary conditions, we just write $h(k; B,\Ef)$ and $E_n(k; B, \Ef)$ instead of $h_{U}(k)$ and $E_n(k)$. For later convenience, we will also recast all the parameters in terms of the dimensionless quantities
\begin{equation}\label{eq-kBE}
 \kad \equiv kW\,,\qquad \Bad \equiv \frac{eBW^2}{\hbar}\,,\qquad \Ead \equiv \frac{em\Ef W^3}{\hbar^2}\,,\qquad \epsilon\equiv \frac{2mEW^2}{\hbar^2}\,.
\end{equation}

Since $h(k; B,\Ef)$ and $h(k-m\Ef/(eB); B,0)$ only differ by a constant term, every energy band in the presence of a non-vanishing electric field is formally linked to the corresponding band at $\Ef=0$ via the relation
\begin{equation}\label{eq-EnElectic}
E_n(k; B,\Ef)=E_n\Biggl(k -\frac{m\Ef}{\hbar B}; B, 0\Biggr)+\hbar k\frac{\Ef}{B}-\frac{m}{2}\frac{\Ef^2}{B^2}\,,
\end{equation}
which, in the rescaled units \eref{eq-kBE}, reads
\begin{equation}\label{eq-EnElecticad}
\epsilon_n(\kad; \Bad,\Ead)=\epsilon_n\Biggl(\kad -\frac{\Ead}{\Bad}; \Bad, 0\Biggr)+2\kad\frac{\Ead}{\Bad}-\frac{\Ead^2}{\Bad^2}\,.
\end{equation}
This means that, as long as we apply the same boundary conditions to each fiber $h(k)$, the spectrum of the system in the presence of a non-vanishing electric field $\Ef$ can be \emph{exactly} computed in terms of the non-electric spectrum via Eqs.~\eref{eq-EnElectic}--\eref{eq-EnElecticad}.

When the electric field vanishes, another interesting relation can be found: introducing  the unitary operator $\mathcal{P}_y$ which inverts the sign of the coordinate $y$, we formally have
\begin{equation}
\mathcal{P}_y h(k; B, 0)\mathcal{P}_y^{\dagger}=h(-k; B,0)\,.
\end{equation}
Therefore, if $\sigma_x U(k)\sigma_x=U(-k)$, $\sigma_x$ being the first Pauli matrix, then $\mathcal{P}_y\dom(h_{U}(k))=\dom(h_{U}(-k))$, and the non-electric spectrum does not depend on the sign of $k$: 
\begin{equation}
	E_n(k; B, 0)=E_n(-k; B, 0)\,.
\end{equation}

\subsection{Asymptotic behavior of the spectrum at zero electric field}
A little more can be said in general about the non-electric fiber $h(k; B, 0)$, observing that it can be interpreted as a two-mode system~\cite{GRD06}. When $\Ef=0$, the maximal value $E\sub{c}(k)$ of the harmonic potential $V_k(y)$ inside the cavity is given by the expression
\begin{equation}
E\sub{c}\equiv \frac{1}{8}m\omega^2_B W^2 \Biggl(1+2\frac{\left|k\right|l^2_B}{W}\Biggr)^2\,.
\end{equation}
Heuristically, such an energy is expected to play the role of a critical energy in the asymptotic behavior of the spectrum  $\sigma(h_{U}(k))$, so that eigenvalues much greater than $E\sub{c}$ ``decouple'' from those lying at the bottom of the spectrum. In the high-energy regime ($E\gg E\sub{c}$) the eigenvalues, as well as their corresponding eigenstates, are indeed expected to be nearly independent of the harmonic potential, thus being uniquely determined by the boundary conditions on the cavity. As such, they can be perturbatively evaluated from the ones of a free particle in a box with the inherited boundary condition. For a related example, see~\cite{GRD06}.

\begin{figure}[tb]
\centering
\includegraphics[width=.9\textwidth]{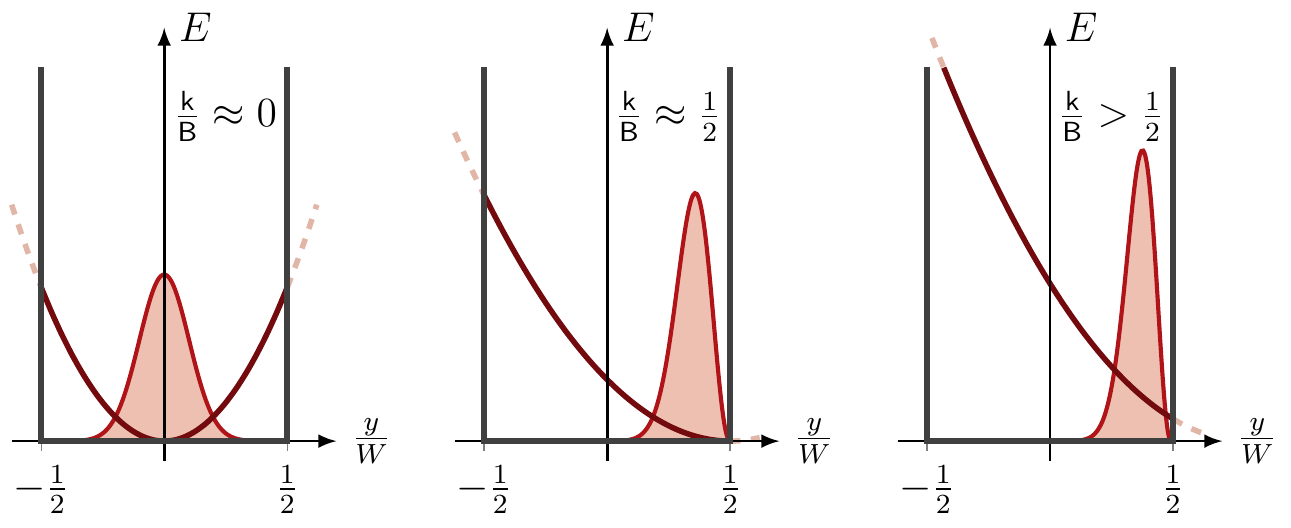}
\caption{Qualitative behavior (in the absence of the electric field $\Ef$) of the effective potential $V_k(y)$, and of a corresponding wavefunction, for different values of the dimensionless quantity $\kad/\Bad$.}
\label{fig-HOwells}
\end{figure}

On the other hand, eigenvalues in the low-energy regime $E\ll E\sub{c}$ may be expected not to “feel” the boundary conditions, since the wavefunctions are confined far from the boundary by the harmonic potential (when the latter is sufficiently narrow). This regime is actually more involved, since, in this case, $\sigma(h_{U}(k))$ heavily depends on both $k$ and $B$, as can be seen by inspecting Eq.~\eref{eq-Hk2}. Some remarkable particular cases, pictorially sketched in Figure~\ref{fig-HOwells}, are discussed in the following.
\begin{itemize}
	\item When $|\kad/\Bad|\ll 1/2$, \ie when the minimum of $V_k(y)$ lies in the middle of the interval, the lowest eigenvalues are close to the ones of an unconstrained harmonic oscillator:
	\begin{equation}\label{eq-spectrum-HO}
	E_n\approx E_n^{\mathrm{HO}}\equiv \hbar\omega_B \Biggl(n+\frac{1}{2}\Biggr)\,,\qquad n\in \mathbb{N}\,.
	\end{equation}
	This holds independently of the boundary condition, since the potential $V_k(y)$ confines the less energetic states far from the boundary. The number $N$ of such “nearly harmonic states” (which may be as well zero!) depends on the the value of $\Bad$: it can be estimated by using $E_n^{\mathrm{HO}}\lesssim E\sub{c}$, from which it follows $N \lesssim \frac{1}{8}(\Bad-4)$.
	\item When $|\kad/\Bad|\approx 1/2$, \ie when the minimum of $V_k(y)$ lies close to an edge of cavity, the lowest eigenvalues are asymptotic to the ones of ``half'' a harmonic oscillator, whose spectrum coincides, for some boundary conditions, with a subset of $\{E_n^{\mathrm{HO}}\}_{n\in \mathbb{N}}$: for instance, for Dirichlet and Neumann boundary conditions (to be introduced later in the text), it is given by Eq.~\eref{eq-spectrum-HO} with $n$ restricted to odd and even integers, respectively. 
	\item Finally, when  $|\kad/\Bad|\gg 1/2$, the part of $V_k(y)$ inside the square well is an arc of parabola that can be approximated by a straight line, at least near the boundary: in this case the asymptotic spectrum approximatively corresponds to the Airy spectrum~\cite{VS04}, \ie the spectrum of a Schr\"{o}dinger operator with a linear potential.
\end{itemize}
Naturally, for arbitrary values of  $k$ and $B$ and for intermediate energies $E\approx E\sub{c}$, we expect the energy levels to interpolate between these asymptotic regimes, with the particular features of such an interpolation depending on the particular choice of boundary conditions.

\section{General transport properties}\label{sec-transport}
Another class of properties of the Hall Hamiltonian that can be discussed in generality are its transport properties, which will be analyzed in this section. In Subsection~\ref{subsec-vel} we introduce the quantum observable describing the \emph{group velocity} (and hence the \emph{electric current}) associated with the eigenstates, while in Subsection~\ref{subsec-hallcon} we describe how to evaluate the \emph{Hall conductivity} of a quantum system.

\subsection{Velocity operator}\label{subsec-vel}
Let us momentarily consider our problem at the classical level. Since the system is confined in the $y$ direction and we are interested in local boundary conditions (as we will discuss in the next section), there can be no net transversal current; accordingly, we only have to examine the \emph{longitudinal} current $I_x(x)$ flowing across a transversal section of the strip. Moreover, since with the chosen boundary conditions the system is invariant under longitudinal translations, such  current does not depend on the $x$ coordinate. For our purposes it is convenient to study the longitudinal velocity $v_x$, which is related to the current $I_x$ by 
\begin{equation}
I_x=-en\sub{sur}W v_x\,,
\label{eq-currentvelocity}
\end{equation}
where $n\sub{sur}$ denotes the surface electron density. Applying the opportune Hamilton equation to (the classical counterpart of)  $H$ reveals the well-known relation, in the Landau gauge \eref{eq-Landau}, between the kinematic momentum $m v_x$ and the canonical momentum $p_x$:
\begin{equation}
m v_x= p_x-eBy\,.
\end{equation}
This relation, which is independent of the electric field $\Ef$, is the starting point to define the longitudinal velocity at the quantum level, \ie as an operator on the Hilbert space $\Leb^2(\Os)$. Let us introduce the operator
\begin{equation}\label{eq-vx}
	v_x\equiv -\frac{1}{m}\Biggl(\iu\hbar\frac{\partial}{\partial x}\otimes 1 + 1\otimes eBy\Biggr),
\end{equation}
initially acting on the minimal domain $\dom_0=\Schw(\R)\otimes\Cinf(I_W)$ introduced in Eq.~\eref{eq-dom0}. By following an approach analogous to that in Section~\ref{sec-model}, we can decompose (the unique self-adjoint extension of) $v_x$, up to a unitary transformation, as the direct integral
\begin{equation}\label{eq-vxint}
\hat{v}_{x} \equiv \Fx v_x \Fx^{-1}= \int_{\R}^{\oplus} v(k)\dd{k}\,,
\end{equation}
the fiber $v(k)$  simply being the multiplication operator
\begin{equation}
v(k)\equiv \frac{1}{m}(\hbar k -eBy)\,.
\end{equation}
Indeed, since each fiber $v(k)$ is bounded and symmetric, it can be automatically extended to a self-adjoint operator on the whole $\Leb^2(I_W)$. Moreover, the function $k\mapsto (v(k)+\iu)^{-1}$ is continuous and hence measurable. Accordingly, by an analogous argument as in Proposition~\ref{prop-dirintbis}, the direct integral operator $\hat{v}_{x}$ is shown to be well-defined and self-adjoint on the domain
\begin{equation}
\dom(\hat{v}_{x} )=\Biggl\{\psi\in \Leb^2(\hat{\R}; \Leb^2(I_W)): \int_{\R} \|v(k)\psi(k)\|^2_{\Leb^2(I_W)}\dd{k}<\infty\Biggr\}\,.
\end{equation}
In order to investigate the transport properties of the system, we are interested in computing the expectation values of $\hat{v}_x$, and thus of each of its fibers, on each eigenfunction $\psi_{n}(k; y)$ corresponding to the eigenvalue $E_n(k)$ of $h_{U}(k)$:
\begin{equation}
v_n(k)\equiv \Braket{\psi_{n}(k; \cdot)|v(k)\psi_{n}(k; \cdot)}_{\Leb^2(I_W)}\,.
\end{equation}
For the sake of simplicity, we slightly tighten our assumptions. Let us assume that:
\begin{itemize}
	\item the family of operators $\bigl\{h_{U}(k)\bigr\}_{k\in\R}$ is defined on a common domain, say $\dom_{U_0}$: this condition holds if we choose the same self-adjoint extension (that is the same boundary condition) for each fiber $h(k)$, setting thus $U(k)=U_0$ for all $k\in\mathbb{R}$;
	\item each eigenvalue $E_n(k)$ is simple, \ie non-degenerate. 
\end{itemize}
An immediate calculation shows that the following equation holds in the weak sense:
\begin{equation}\label{eq:dHk/dk}
\frac{1}{\hbar}\frac{\dd h_{U_0}(k)}{\dd k}=v(k)\,,
\end{equation}
that is, for all $\phi,\psi\in\dom_{U_0}$,
\begin{equation}
\frac{1}{\hbar}\frac{\mathrm{d}}{\mathrm{d}k}\Braket{\phi|h_{U_0}(k)\psi }_{\Leb^2(I_W)}=\Braket{\phi\big|v(k)\psi}_{\Leb^2(I_W)}\,.
\end{equation}
Therefore, under the above two conditions, the Hellmann-Feynman equation holds~\cite{EFC,ZG} and we can write\footnote{We mention however that a generalized Hellmann-Feynman formula exists also for degenerate spectra and, in some cases, when each fiber $h(k)$ is defined on a different domain.}
\begin{equation}\label{eq-vnk}
v_n(k)=\frac{1}{\hbar}\frac{\dd E_n}{\dd k} =\frac{1}{\hbar}E_n'(k)\,.
\end{equation}
The relation~\eref{eq-EnElectic} between the electric spectrum $E_n(k; B,\Ef)$ and the non-electric one can be applied also to the group velocity $v_n(k)$, yielding:
\begin{equation}\label{eq-vnkElectric}
v_n(k; B, \Ef)=v_n\Biggl(k-\frac{m\mathcal{E}}{\hbar B}; B, 0\Biggr)+\frac{\Ef}{B}\,.
\end{equation}
Interestingly, the mean velocity in the presence of the electric field only differs by a constant term, which thus coincides with a constant \textit{drift velocity}  in agreement with classical mechanics.

We conclude by observing that one may also define the \emph{local} velocity operator 
\begin{equation}
v_x(y_0)\equiv v_x\,\ddelta(y-y_0)\,,
\end{equation}
where $v_x$ is the operator of Eq.~\eref{eq-vx}. Repeating the same steps as above one finds the corresponding fiber $v(k; y_0)\equiv \frac{1}{m}(\hbar k -eBy)\ddelta(y-y_0)$, so that its expectation values are given by
\begin{eqnarray}
v_n(k; y_0)  &\equiv\braket{\psi_{n}(k; \cdot)|v(k; y_0)\psi_{n}(k; \cdot)}_{\Leb^2(I_W)}\\
&=-\omega_B (y_0-kl^2_B)|\psi_n(k; y_0)|^2\,.\label{eq-localvel}
\end{eqnarray}
Note that this quantity vanishes either when $y_0=kl_B^2$ (which corresponds to the center of the non-electric harmonic potential, see Figure~\ref{fig-costrainedHO}) and in correspondence of the wavefunction nodes.

\subsection{Hall conductivity}\label{subsec-hallcon}
Let us start  by considering again the problem at the classical level. For a Hall system, the conductivity tensor $\sigma$, defined by the expression
\begin{equation}\label{eq-sigma}
	\vb{J}=\sigma\Eb
\end{equation} 
with $\vb{J}=-en\vb{v}$ being the current density, can be classically evaluated using the Drude theory~\cite{Yo02}. In particular, when there is no scattering (adiabatic limit), its diagonal components $\sigma_{xx}$ and $\sigma_{yy}$ vanish while the \emph{Hall conductivity} $\sigma_{xy}$, which relates the longitudinal current $J=J_x$ with the transversal electric field $\Ef=\Ef_y$, is given by
\begin{equation}\label{eq-sigmaxyD}
\sigma_{xy}=\frac{ne}{B}\,.
\end{equation}
This classical expression can be adapted for a quantum system by substituting the macroscopic electron density $n$ with the number of states per unit volume up to a certain Fermi energy $E\sub{F}$, that is, the \emph{cumulative density of states}. Its expression, in the zero-temperature limit which we are interested in, is given by the expression 
\begin{eqnarray}
\hskip-40pt
\frac{N(E\sub{F}; B, \Ef)}{V}&\equiv
\frac{1}{V}\sum_{n=0}^{+\infty}\int_{-\infty}^{E\sub{F}}\dd{E}\int_{-\infty}^{+\infty}\ddelta(E-E_{n}(k; B, \Ef))\frac{L}{2\pg}\dd{k} \\ 
&=\frac{1}{V}\sum_{n=0}^{+\infty}\int_{-\infty}^{+\infty}\htheta(E\sub{F}-E_{n}(k; B, \Ef))\frac{L}{2\pg}\dd{k}\label{eq-N}\\
&=\frac{1}{V} \sum_{n=0}^{+\infty} \int_{-\infty}^{+\infty}  \htheta\Biggl(E\sub{F}-\hbar k\frac{\Ef}{B}-\frac{m}{2}\frac{\Ef^2}{B^2}-E_n(k; B,0)\Biggr)\frac{L}{2\pg}\dd{k}  \,,
\end{eqnarray}
where $\uptheta(x)$ denotes the Heaviside step function. Note that the length $L$, which is formally infinite for the strip $\Os$, actually simplifies with the volume $V$ leaving the (finite) cross section $A=V/L$. 

For our purposes, it will be convenient to consider the \emph{conductance} 
\begin{equation}
G_{xy}\equiv \frac{A\sigma_{xy}}{W}\,,
\end{equation}
which differs from $\sigma_{xy}$ merely by a geometric factor but does not depend on the cross section $A$. Recall that the unit of measure of $G_{xy}$ is the conductance quantum $e^2/h$, the latter being in turn  the reciprocal of the von Klitzing constant $R\sub{K}\approx \SI{2.58}{\kilo\ohm}$. By plugging the quantum density \eref{eq-N} in the classical expression \eref{eq-sigmaxyD}, we readily obtain
\begin{eqnarray}
\hskip-40pt
G_{xy}(E\sub{F}; B, \Ef)&=\frac{e^2}{h}\sum_{n=0}^{+\infty}\frac{\hbar}{eBW}\int_{-\infty}^{+\infty} \htheta\Biggl(E\sub{F}-\hbar k\frac{\Ef}{B}-\frac{m}{2}\frac{\Ef^2}{B^2}-E_n(k; B,0)\Biggr)\dd{k} 
\\
	&=\frac{e^2}{h}\sum_{n=0}^{+\infty}\int_{-\infty}^{+\infty}\htheta(\epsilon\sub{F}-\epsilon_n(\kad; \Bad, \Ead))\frac{\dd{\kad}}{\Bad}\,,\label{eq-Gxy}
\end{eqnarray}
where the second line is in terms of the dimensionless quantities introduced in Eq.~\eref{eq-kBE}.

This simple result can be obtained also by quantizing directly the definition~\eref{eq-sigma}, thus avoiding to refer to the (classical) Drude theory and in particular to the expression~\eref{eq-sigmaxyD}. 
We give some details of this approach in~\ref{sec-conductivity}. 

Interestingly, the quantum conductance $G_{xy}(E\sub{F}; B, \Ef)$ is actually a non-linear function of the applied electric field $\Ef$, since we are not relying on the linear response theory (\ie the celebrated Kubo formula~\cite{ChaPie95}) as it is usually done in many alternative models of the QHE; this full dependence of $G_{xy}$ on $\Ef$ will allow us to predict the breakdown of the QHE at a sufficiently strong electric field, the latter being a well-know experimental phenomenon,  see e.g.~\cite{KaHiNa93, Nach99}.

We point out that a similar result has been already obtained for the simpler case in which the electron is free to move in the whole plane $\Omega=\R^2$~\cite{Kr04,Kr06}, and thus there are no complications involved in the selection of a self-adjoint extension (indeed in such a case the Hall Hamiltonian has  a unique self-adjoint extension on $\Cinf(\R^2)\subset\Leb^2(\R^2)$, see e.g.\ Theorem~2 of~\cite{Lein83}). In these works the authors explain how to generalize the result in order to account for the electron spin, for the scattering and for finite-temperature effects. In~\cite{Kr06}, in particular, they compare their theoretical results with some experiments, obtaining a surprisingly good agreement which thus validates the potentiality of our model. 

The aim of the remaining part of this paper is therefore to complete the analysis of this model, understanding how boundary conditions and  finite-size effects influence the quantization of the Hall conductivity.

\section{Boundary conditions for the quantum Hall effect}\label{sec-saext}

Proposition~\ref{prop-dirintbis} allows us to describe all possible fibered boundary conditions for $H$ on the Hall strip. However, not all such conditions may be relevant for the description of the QHE. As a matter of fact, we will now restrict our attention to boundary conditions satisfying the following requests. 
\begin{enumerate}
	\item We will limit our considerations to fibered boundary conditions which are \emph{independent of $k$}, \ie we will select the same self-adjoint extension for \emph{all} the fiber operators $h(k)$, which will hence be defined, for a given $U_0\in\UU(2)$, on the common domain 
	\begin{equation}\label{eq-bc}
	\dom_{U_0}=\{\psi\in H^2(I_W) :  \iu(I+U_0)\Psi=(I-U_0)\Psi'\}\,,
	\end{equation}
	see Eqs.~\eref{eq-boundarydata}--\eref{eq-Ddagger}. However, we mention that boundary conditions with non-trivial dependence on $k$ may play a role for the description of the QHE~\cite{JJV95, AANS98}, see also~\ref{sec-position} for a simple example.  \label{cond-trivial}
	\item We will exclude all \emph{non-local} boundary conditions, \ie we will only consider boundary conditions not mixing the values of the wavefunction and its derivative at the two edges $y=\pm W/2$, thus keeping the original topology of the strip $\Os$. See~\cite{NiuTho87}, however, for other interesting models of the QHE in which non-local boundary conditions are employed. \label{cond-local} 
	\item We will only consider boundary conditions which preserve the symmetry under the transversal reflection $\mathcal{P}_y$ (mapping $y$ to $-y$); from the physical point of view this should be understood as a property of the Hall device itself, \ie we are assuming that both its edges are made of, say, the same material; by contrast, the whole Hall system \eref{eq-Omegas}--\eref{eq-BEfields} is actually not $\mathcal{P}_y$-invariant, as $\mathcal{P}_y H(B,\Ef) \mathcal{P}_y^{-1}=H(-B,-\Ef)$.\label{cond-y}
\end{enumerate}
These requests greatly reduce the number of free parameters describing the allowed boundary conditions on the strip. The first assumption~\eref{cond-trivial} guarantees that all fibers satisfy the same boundary condition, accordingly associated with a single unitary matrix $U_0$, as in Eq.~\eref{eq-bc}. The locality condition~\eref{cond-local} implies that $U_0$ is a diagonal matrix, so that the components of $\Psi$ and $\Psi'$ in Eq.~\eref{eq-Uboundary} do not mix:
\begin{equation}\label{eq-unit}
U_0=\left(\begin{array}{@{}cc@{}}
\e^{\iu\theta_1} & 0 \\ 0 & \e^{\iu\theta_2}
\end{array}\right)
\,,\qquad\theta_1,\theta_2\in\left[0,2\pg\right)\,.
\end{equation}
Finally, the last condition~\eref{cond-y} can be readily shown to be satisfied if and only if $\theta_1=\theta_2\equiv \theta$. As a result, the assumptions~\eref{cond-trivial}--\eref{cond-y} are only satisfied by a one-parameter group of fibered boundary conditions, indexed by a parameter $\theta\in\left[0,2\pg\right)$, which for each $k\in\mathbb{R}$  reads
\begin{equation}\label{eq-RobinTheta}
\iu(1+\e^{\iu\theta})\psi(k; \pm)=\pm l_0(1-\e^{\iu\theta})\psi'(k; \pm)\,,
\end{equation}
where for compactness we have set 
\begin{equation}
\psi^{(\prime)}(k; \pm)\equiv\psi^{(\prime)}(k; \pm W/2)
\end{equation}
and where $\psi(k; y)\in H^2(I_W)\subset \Leb^2(I_W)$ as in Eq.~\eref{eq-eigeneq}. Such conditions are well-known in the literature, as we now explain with some details. 

\subsection{Fibered Robin boundary conditions}
It is customary to exploit the Cayley transform, a well-known bijection between the unit circle $\mathbb{S}_1$ 
 and the extended real line $\overline{\R}=\R\cup\{\infty\}$, namely
\begin{equation}\label{eq-alpha}
\theta\in\left[0,2\pi\right)\mapsto \alpha\equiv \iu\frac{1}{l_0}\frac{1+\e^{\iu\theta}}{1-\e^{\iu\theta}}=-\frac{1}{l_0}\cot\left(\case{\theta}{2}\right)\in\overline{\mathbb{R}}\,,
\end{equation}
in order to rewrite Eq.~\eref{eq-RobinTheta} more compactly as 
\begin{equation}\label{eq-Robin}
\psi'(k; -)=-\alpha\, \psi(k; -)\qquad\textnormal{and}\qquad
\psi'(k; +)=+\alpha\, \psi(k; +)\,.
\end{equation}
This condition represent a special case of a \emph{Robin} (or mixed) boundary condition on the segment, the most general case being the one with two different parameters $\alpha_\pm$ on the two edges; $\alpha$ will be called the \emph{Robin parameter}. Before going on, let us mention that these fibered Robin conditions do actually correspond to \emph{global} Robin conditions in the position representation, see~\ref{sec-position}.

Two peculiar cases are obtained for $\alpha=\infty$ or $\alpha=0$: in the former case, Eq.~\eref{eq-Robin} reduces to
\begin{equation}
\psi(k; -)=0\qquad\textnormal{and}\qquad
\psi(k; +)=0\,,
\end{equation}
that is, to a \emph{Dirichlet} boundary condition at both the extremal points of $I_W$. In the latter case, we get
\begin{equation}
\psi'(k; -)=0\qquad\textnormal{and}\qquad
\psi'(k; +)=0\,,
\end{equation}
corresponding to a \emph{Neumann} boundary condition. While both Dirichlet and Neumann boundary conditions are just special cases of Robin conditions, we will often stress their role by writing explicitly all equations in the cases $\alpha=\infty$, and $\alpha=0$ together with the general case. 

The sign of $\alpha$ has an interesting physical interpretation, when linked to the corresponding eigenfunctions: Robin conditions either ``repel" or ``attract" wavefunctions at the boundary respectively in the cases $\alpha < 0$ or $\alpha \ge 0$, as pictorially shown in Figure~\ref{fig-Robin}. In particular, we will henceforth focus on three different values of the dimensionless parameter $\alpha W$, that is:
\begin{itemize}
	\item $\alpha W=\infty$, \ie Dirichlet boundary condition;
	\item $\alpha W=0$, \ie Neumann boundary condition;
	\item $\alpha W=10$, \ie an intermediate case between Dirichlet and Neumann which, for the sake of simplicity, we will simply refer to as \emph{the} Robin boundary condition.
\end{itemize}
Our choice for the third boundary condition reflects the fact that positive Robin conditions are known to enhance the appearance of \emph{edge states}, \ie states that are localized in a neighborhood of the boundary: these states, at least in the absence of the magnetic field $B$, are associated with negative energy levels. Conversely, Robin conditions with a negative parameter merely constitute an interpolation between Dirichlet and Neumann conditions, apparently without any additional feature.

\begin{figure}[tb]
	\centering
\includegraphics{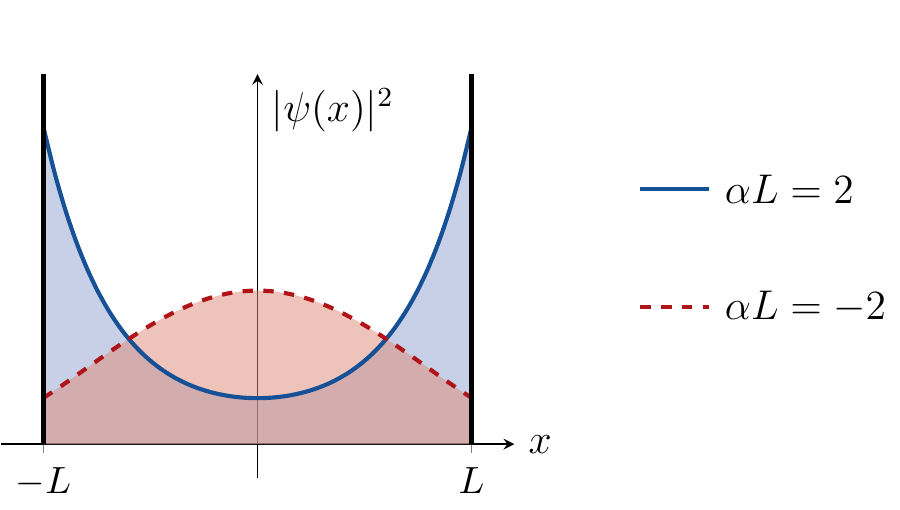}
\caption{Squared modulus of some ground-state eigenfunctions of the free Laplacian on the segment line $\left(-L,L\right)$, subjected to Robin boundary conditions $\psi'(\pm L)=\pm \alpha\, \psi(\pm L)$; it can be seen that the boundary conditions have a repulsive effect when $\alpha<0$ and an attractive one  when $\alpha\ge 0$.}
	\label{fig-Robin}
\end{figure}

Now, in order to concretely obtain the energy levels and the associated eigenfunctions, we have to solve the eigenvalue problem  \eref{eq-eigeneq} for each boundary condition; we will also omit the dependence of the energy $E(k)$ on the wave-number $k$ where no ambiguity rises.  By appropriately rescaling the variables, the eigenvalue equation can be easily recast as the Weber differential equation: see~\ref{sec-weber} for details. For every energy $E$, this equation admits a two-dimensional space of solutions, any pair of independent generators being expressible in terms of the confluent hypergeometric function: by denoting the elements of such a pair as $u^{1}$ and $u^2$, the general solution of Eq.~\eref{eq-eigeneq} for an arbitrary $E$ is therefore given by
\begin{equation}
\psi_{E}(k; y)\equiv c_1 u_E^1( k; y) + c_2 u_E^2(k; y)\,,\qquad c_1,c_2\in \mathbb{C}\,.
\end{equation}
Three constraints have to be imposed on this expression:  the normalization constraint
\begin{equation}
\int_{I_W} |\psi_{E}(k; y)|^2 \dd{y}=1\,
\end{equation}
and the boundary condition at the two edges of the segment line. As a consequence, $c_1$ and $c_2$ are fixed (up to a global phase) and the admissible energy values are quantized, as they correspond to the real roots of a suitable \emph{spectral function} $F_k(E)$~\cite{AIM15}. A long but straightforward calculation yields explicit expressions for $\psi_{E}(k; y)$ (up to a normalization constant) and of $F_k(E)$ for each boundary condition:
\begin{eqnarray*}
\psi_{E}^{\mathrm{D}}(k; y)&=u_E^2(k; -) u_E^1(k; y) - u_E^1(k; -) u_E^2(k; y) 	\qquad&\textnormal{(Dirichlet)}\\[3pt]
\psi_{E}^{\mathrm{N}}(k; y)&={u_E^2}'(k; -) u_E^1(k; y) - {u_E^1}'(k; -) u_E^2(k; y)&\textnormal{(Neumann)}\\[3pt]
\psi_{E}^{\mathrm{R}}(k; y)&= C_E^{2}(k, \alpha) u_E^1(k; y) - C_E^{1}(k, \alpha) u_E^2(k; y)&\textnormal{(Robin)} 
\end{eqnarray*}
where $C_E^{i}(k, \alpha)\equiv  {u_E^i}'(k; -)+\alpha u_E^i(k; -)$. The corresponding spectral functions are
\begin{eqnarray*}
F_k^{\mathrm{D}}(E)&=u_E^2(k; -) u_E^1(k; +) - u_E^1(k; -) u_E^2(k; +) \qquad &\textnormal{(Dirichlet)}\\[3pt]
F_k^{\mathrm{N}}(E)&={u_E^2}'(k; -) {u_E^1}'(k; +) - {u_E^1}'(k; -) {u_E^2}'(k; +)
&\textnormal{(Neumann)}\\[3pt]
F_k^{\mathrm{R}}(E)&=f_k^{+}(E, \alpha)-f_k^{-}(E,\alpha)
&\textnormal{(Robin)} 
\end{eqnarray*}
where
\begin{equation}
f_k^{\pm}(k,\alpha)\equiv 
[{u_E^1}'(k; \mp) \pm \alpha u_E^1 (E, k; \mp)] 
[{u_E^2}'(k; \pm) \mp  \alpha u_E^2 (k; \pm)]\,.
\end{equation}
Since the equation $F_k(E)=0$ is generally transcendental (and involve some special functions), its solutions have to be evaluated numerically, and this will be the aim of the next subsection.

\begin{figure}[t]
\centering
\includegraphics{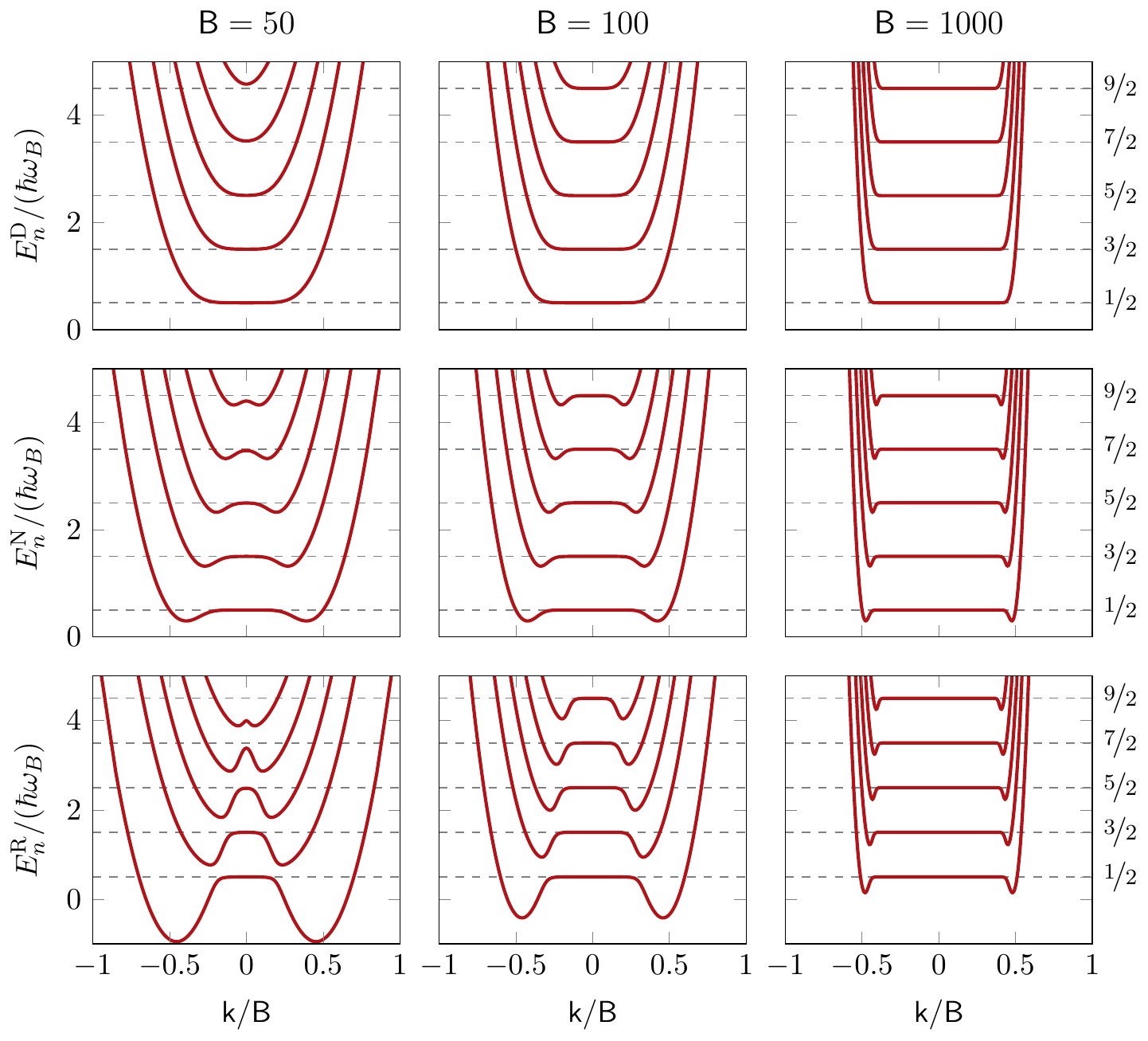}

\caption{Dispersion diagrams in the rescaled variables $E/(\hbar\omega_B)$ and $\kad/\Bad$ for different values of the magnetic field and for different boundary conditions; from left to right: $\Bad=50,100,1000$; from top to bottom: Dirichlet, Neumann and Robin boundary conditions.}
\label{fig-dispersion}
\end{figure}

\subsection{Dispersion diagram}\label{sec-spectrum}
As already discussed, the spectrum of $H_{U}$ generally consists of many energy bands depending on the wave-number $k$; these bands can be obtained by finding the roots of the spectral function $F_k(E)$ for different values of $k$. The result of this operation, in the non-electric case ($\Ef=0$) and for some values of the magnetic field $B$, is shown in Figure~\ref{fig-dispersion} for all the above-mentioned boundary conditions. In this figure we plotted the so-called \emph{dispersion diagrams}, which show the dependence of the energy levels $E_n(k)$ on the wave-number $k$; the dashed lines represent the Landau levels $E_n^{\textup{HO}}=\hbar\omega_B(n+1/2)$, which constitute the spectrum of $H(B,\Ef=0)$ when the electron is free to move in the whole plane. For completeness, in Figure~\ref{fig-eigen} we also plotted some of the eigenfunctions corresponding to the first energy band. 

\begin{figure}[t]
\centering
\includegraphics{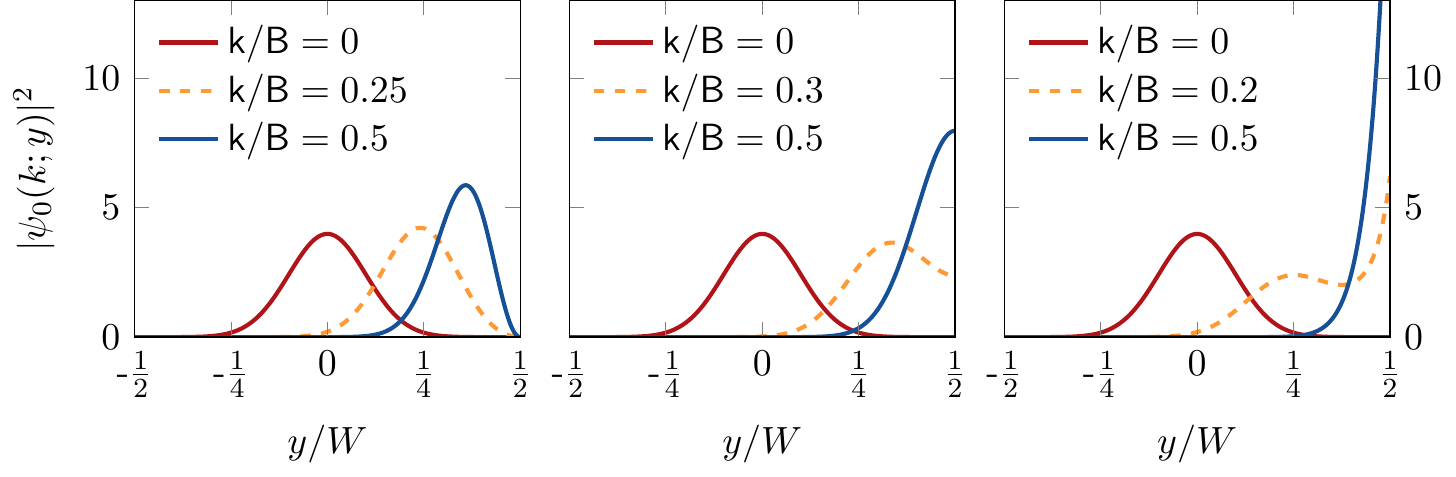}
\caption{Squared modulus of the eigenfunctions $\psi_{0}(k; y)$ belonging to the first energy band, for $\Bad=50$ and for different values of $\kad$; from left to right: Dirichlet, Neumann and Robin boundary conditions.}
\label{fig-eigen}
\end{figure}

Different features of these dispersion diagrams have long been qualitatively known~\cite{Halp82, Butt88}; to the best of our knowledge, however, only the dispersion diagram for Dirichlet conditions was quantitatively computed~\cite{Ve11, MU17}. In particular, Neumann and Robin spectra display interesting new features,  related to the \emph{edge currents}, which, to our knowledge, have never been accounted for in the existing literature. We list here some comments.

\begin{itemize}
	\item If the magnetic field is sufficiently high, the energy bands have a common behavior which is independent of the boundary conditions: they are flat for $|\kad/\Bad|\lesssim 1/2$, mimicking the Landau levels, but rise when $|\kad/\Bad|\gtrsim 1/2$. The role of the boundary condition becomes indeed relevant only at intermediate values of $k$, \ie when $|\kad/\Bad|\approx 1/2$.	Physically, eigenvalues with small values of $k$, \ie small momenta, depend negligibly on the behavior at the boundary and are hence close to the corresponding Landau levels. As shown in Figure~\ref{fig-eigen}, this property has an immediate counterpart at the level of the eigenstates: the ones with small $k$ are largely concentrated in the bulk, whereas states with large momentum concentrate near one edge of the strip. For this reason, states with small and large values of $k$ (or better of $\kad/\Bad$) will be respectively referred to as \emph{bulk} and \emph{edge} states.
	
\item To understand the raising of the energy bands at large values of $k$, we can think of a classical gas in a box: when the volume of the box is reduced, the energy of the gas increases. Note that this argument is not just a classical analogy: as showed explicitly in~\cite{DiAF13}, if the wall of the box are moved the energy of the system changes also at the quantum level.
In our case, an electron in the bulk is contained in an effective box of length $l_B$; when instead $|\kad/\Bad|\approx 1/2$, which means that the center of the quadratic potential lies near an edge (recall Figure~\ref{fig-HOwells}), the box width shrinks to $l_B/2$, and it is further reduced when $|\kad/\Bad|$ goes over $1/2$. This localization effect is again reflected by the shape of the eigenfunctions.
	
\item As for the dependence of the dispersion diagrams on the magnetic field $B$, the width of the flat region (bulk) of the diagram increases as $B$ increases, \ie as the magnetic length $l_B=\sqrt{\hbar/eB}$ decrease. Remarkably, this behavior is again independent of the boundary conditions.

\item In the case of Neumann and (positive) Robin conditions, a novel phenomenon is visible in the dispersion diagrams: before rising in correspondence of $|\kad/\Bad|\approx 1/2$, the energy bands $E_n(k)$ fall developing a characteristic negative bump, its width being again related to the magnetic length $l_B$. Interestingly, because of these bumps, the ground energy of the system is no longer given by $E_0(k=0)\approx \hbar\omega_B/2$ as in the case of the Dirichlet spectrum. This energy lowering may be due to the fact that Neumann and (positive) Robin conditions do not repel the wavefunctions from the boundary, unlike the Dirichlet condition. 
As we will argue in the following section, the states corresponding to (a part of) the bumps have properties which are halfway between those of bulk and edge states, and have thus no classical analogues.

\item The Robin boundary conditions interpolate between Neumann ($\alpha=0$) and  Dirichlet ($\alpha=\pm \infty$) boundary conditions. The bumps persist for any finite value of $\alpha$. However, despite the limit $\alpha\to - \infty$ is very regular and the bumps are continuously decreasing till they disappear, the limit $\alpha\to + \infty$ is quite singular~\cite{AIM05}. As $\alpha$ grows to $+\infty$ the size of the bumps increases, meaning that edge states reach very high negative energies at the same time that they modify the central part of the bulk spectrum (see Figure~\ref{fig-RobToDir}). The states responsible for this anomalous behavior are very localized at the edge and become delta-like in the extreme limit $\alpha\to +\infty$, which means that they disappear from the spectrum as it corresponds to the Dirichlet case.
\end{itemize}

\begin{figure}[tb]
\centering
\includegraphics{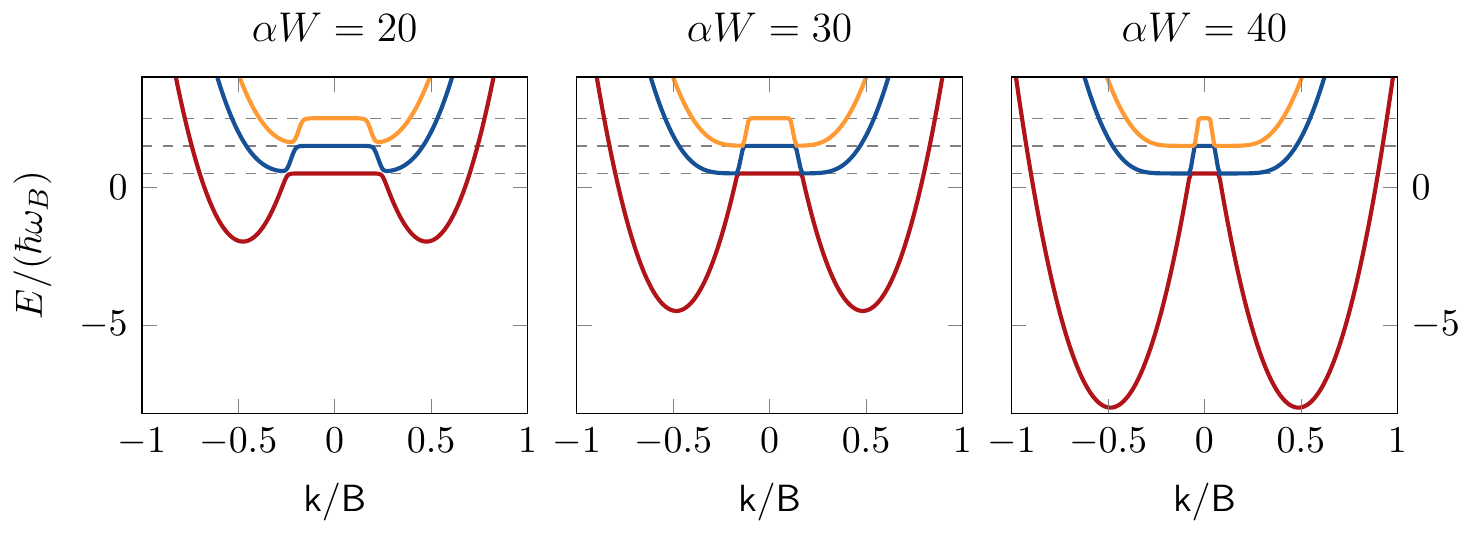}
\caption{Dispersion diagrams in the rescaled variables $E/(\hbar\omega_B)$ and $\kad/\Bad$ for Robin boundary conditions with different values of the parameter $\alpha W$.}
\label{fig-RobToDir}
\end{figure}

Summing up, a Hall strip with local, fibered and symmetry-preserving boundary conditions is characterized by a dichotomy between \emph{bulk states} (with small momentum and mostly localized in the bulk of the strip)  and \emph{edge states} (with large momentum and mostly localized at one edge of the strip). Besides, in the presence of attractive boundary conditions, such as the Neumann and positive Robin ones, a third kind of states appears. To better understand the physical differences between these states, we devolve the next two subsections to the analysis of their transport properties.

\begin{figure}[tb]
\centering
\hspace*{-12pt}
\includegraphics{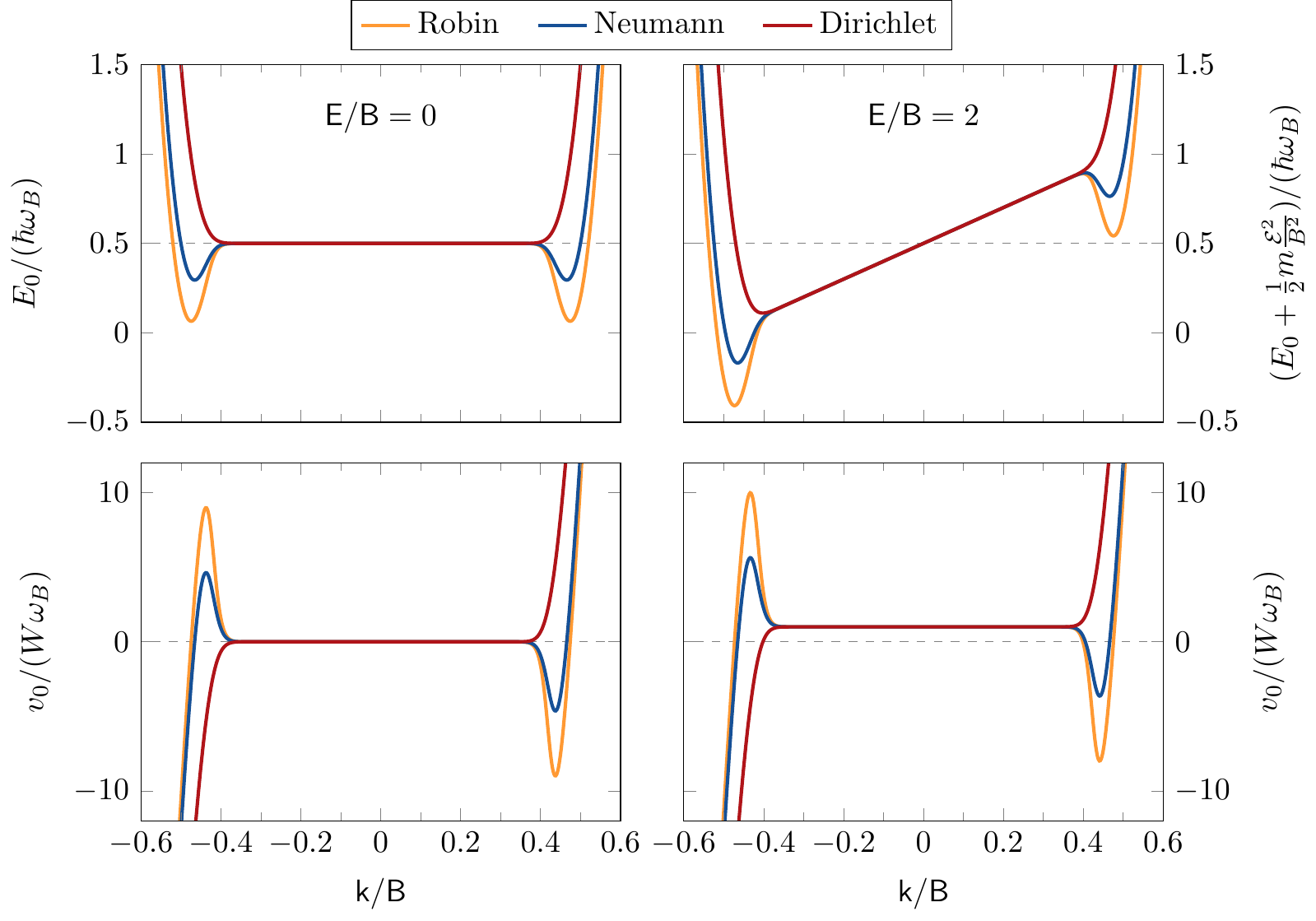}
\caption{Effect of different boundary conditions on the first energy band $E_0(k)$ (top row) and on the associated group velocity $v_0(k)$ (bottom row); left column: $\Bad=500$ and $\Ead=0$; right column: $\Bad=500$ and $\Ead=1000$.}
	\label{fig-velocity}
\end{figure}

\subsection{Bulk-edge dichotomy}\label{sec-be}
In Figure~\ref{fig-velocity} we plotted, both for the cases $\Ef=0$ and $\Ef>0$, the dispersion diagram of the first energy band $E_0(k)$ and its associated group velocity $v_0(k)$, comparing Dirichlet, Neumann and Robin conditions.\footnote{The Robin spectrum is non-degenerate for every choice of the Robin parameter $\alpha$ (see \eg~\cite{LS91}), so the assumptions leading to the Hellmann-Feynman equation~\eref{eq-vnk} actually hold.}

Let us first discuss the velocity diagram when $\Ef=0$, which corresponds to the bottom left panel. Since bulk states are associated with a flat (\ie constant) energy dispersion, these states do not propagate, independently of the boundary conditions; conversely, edge states develop a finite group velocity, either positive or negative, even in the absence of an applied electric field! However, note that, since the non-electric group velocity $v_n(k; B, \Ef=0)$ is an odd function, the \emph{net} electric current vanishes, as expected. 

We note that the edge states propagation is already predicted at the classical level, although only in the case of a non-vanishing electric field, where it is usually explained in terms of skipping orbits at the two edges of the strip~\cite{Yo02}. Interestingly, however, at the quantum level novel phenomena may appear, depending on the boundary conditions. The situation is sketched in Figure~\ref{fig-semiclassical}, where states are represented in a semi-classical manner, and described in detail in the following.
\begin{itemize}
	\item In the Dirichlet case, the direction of propagation depends exclusively on the sign of $\kad/\Bad$, and hence on the sign of $k$; since  for a sufficiently high magnetic field the eigenfunctions with $|\kad/\Bad|\gtrsim 1/2$ are localized on a neighbourhood of $y=W\kad/\Bad$, edge states on the upper part of the strip propagate to the right while those on the bottom one move to the left; as such, we can associate to them a negative (\ie orbital clockwise) \emph{chirality}~\cite{Yo02, AANS98}. Contrarily, bulk states have a positive chirality, as one can deduce by evaluating their local velocity~\eref{eq-localvel}. Therefore, we can conclude that Dirichlet conditions exactly respect the classical bulk-edge dichotomy, \ie bulk and edge states have always opposite chiralities.
	\item The classic dichotomy is \emph{not}, instead, respected for Neumann and Robin conditions, since the corresponding spectra present a new kind of states whose features are intermediate between the bulk and the edge ones. These states belong to the part of the negative bumps that is closest to the bulk: although they 
	have a well-defined group velocity, their chirality is positive.
	Accordingly, carrying on the idea of classifying bulk and edge states solely by means of their chirality, we will denote such states as \emph{improper} bulk states.
\end{itemize}
\begin{figure}[tb]
	\centering
	\includegraphics{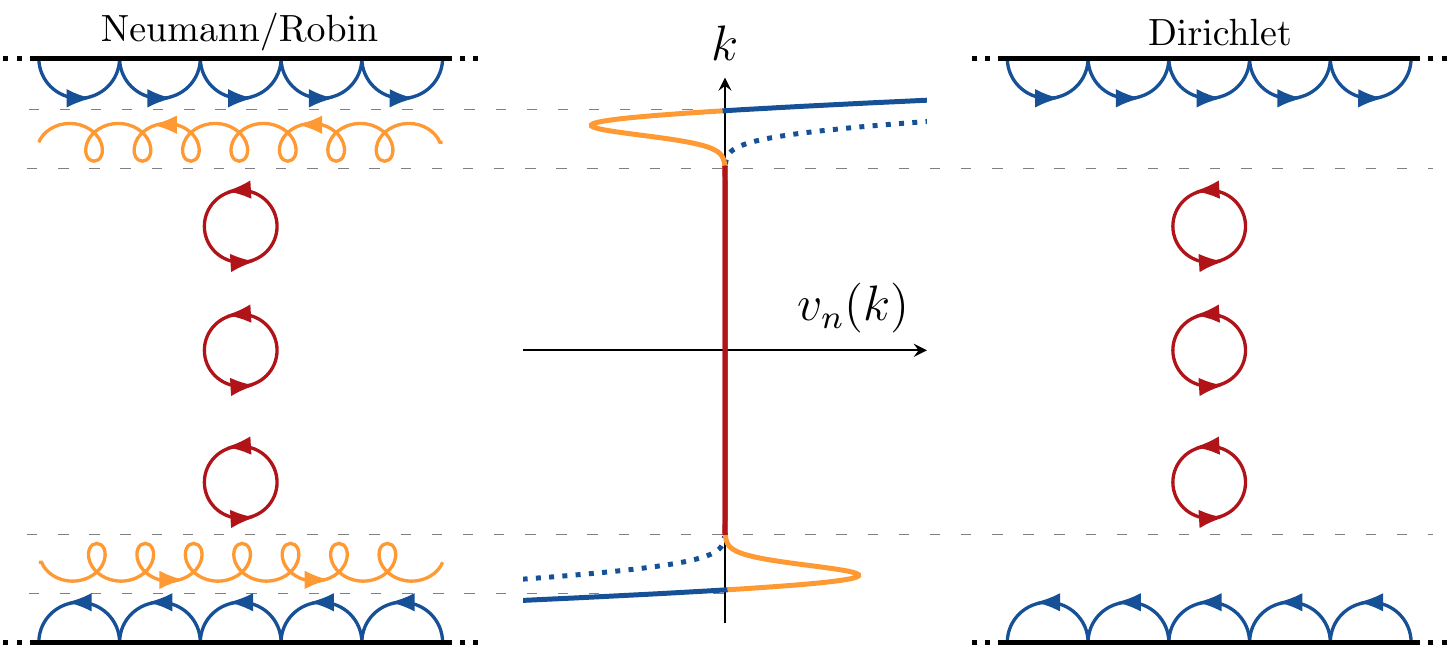}
\caption{Left and right: semi-classical representation of some eigenfunctions; orbit of different colors depict different kind of states: \textcolor{Blu}{\textbf{edge}}, \textcolor{orange!80}{\textbf{improper bulk}} and \textcolor{darkred}{\textbf{bulk}} (see the text). Center: qualitative plot of the group velocity $v_n(k)$ associated with Dirichlet (dotted line) and Neumann/Robin (solid line) boundary conditions.}
\label{fig-semiclassical}
\end{figure}
Neumann and Robin spectra have another interesting feature: they display a clear separation between bulk states (both proper and improper) and edge ones, such a distinction being somewhat arbitrary for the Dirichlet spectrum. Setting $\mathcal{H}=\Leb^2(\Os)$, let us define the $n$-th band space $\mathcal{H}_n$ as the $H$-invariant subspace of $\mathcal{H}$ containing all the (generalized) eigenstates associated with the $n$-th energy band $E_n(k)$. For each energy band we further denote with $k_n$ the (positive) wave-number associated with the minimum of the  negative bump present in the Neumann and Robin spectra. In this way, by defining $\mathcal{H}_{\mathrm{b},n}$ and $\mathcal{H}_{\mathrm{e},n}$ as the subspaces of $\mathcal{H}_{n}$ which are respectively invariant under the Hamiltonians
\begin{equation}
	H_{\mathrm{b},n}\equiv \int_{|k|\le k_n}^{\oplus} h(k) \dd{k} \qquad\text{and}\qquad H_{\mathrm{e},n}\equiv \int_{|k|>k_n}^{\oplus} h(k)\dd{k}\,,
\end{equation}
(with the appropriate boundary conditions) and by setting
\begin{equation}
	\mathcal{H}\sub{b}\equiv \bigoplus_{n=0}^{+\infty}\mathcal{H}_{\mathrm{b},n} \qquad\text{and}\qquad \mathcal{H}\sub{e}\equiv \bigoplus_{n=0}^{+\infty}\mathcal{H}_{\mathrm{e},n}\,,
\end{equation}
we can sharply distinguish the bulk states from the edge ones, that is:
\begin{equation}
	\mathcal{H}=\mathcal{H}\sub{b}\oplus \mathcal{H}\sub{e}\,.
\end{equation}
Consequently, in each subspace the wavefunctions have the same chirality, which is positive for the bulk space $\mathcal{H}\sub{b}$ and negative for the edge space $\mathcal{H}\sub{e}$. Note that a similar decomposition has been proposed also for the Dirichlet extension~\cite{DeBP99}, which however suffers  of the aforementioned problems. See also~\cite{AANS98}, where a clear splitting is induced by introducing  a chiral boundary condition.

Before moving on to the analysis of the Hall conductivity in the next subsection, let us briefly discuss the effects of the transversal electric field $\Ef$ on the spectra and on the eigenstates propagation. As shown in the right panels of Figure~\ref{fig-velocity}, the electric field breaks the $k$-parity of the spectra, so that now the two edges behave differently: in particular, for $\Ef>0$, the energy of the states at the bottom of the strip is lowered, while that of the states at the top of the strip is raised. Furthermore, since the bulk dispersion energy is no longer flat, all bulk states acquire a (positive) drift velocity, as in the classical case. This drift velocity is acquired also by the edge states, and it is added to their intrinsic velocity $v_n(k; B, \mathcal{E}=0)$. As a matter of fact, in this case the group velocity $v_n(k)$ is no longer an odd function, and the net current acquires a non-vanishing value.

\begin{figure}[tb]
\centering
\begin{tikzpicture}
\node at (-7,3.1) {\textbf{(a)}};
\node at (0,0) {\includegraphics{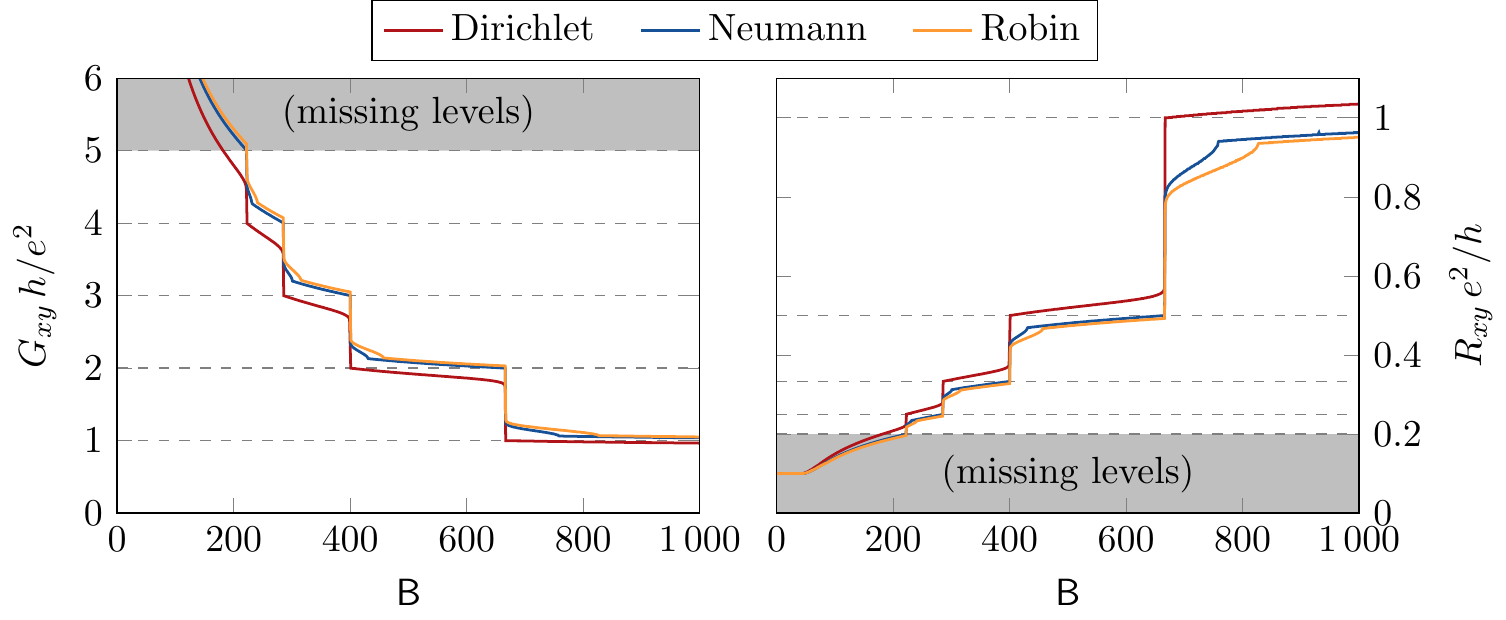}};

\node at (-7,-4) {\textbf{(b)}};
\node at (0,-8.5) {\includegraphics{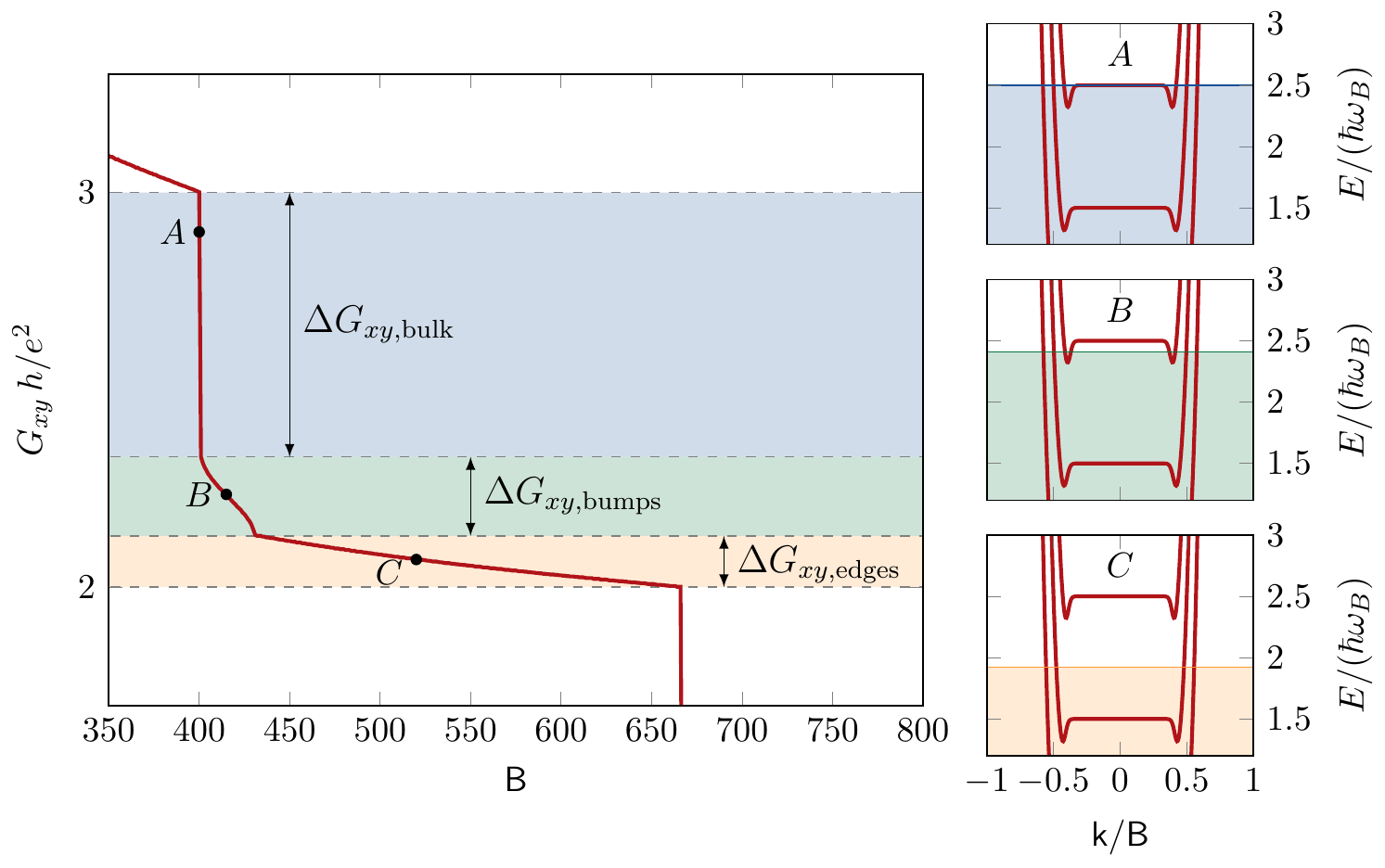}};
\end{tikzpicture}
\caption{(a) Effect of different boundary conditions on the Hall conductance $G_{xy}$ (left) and on the Hall resistance $R_{xy}$ (right), both plotted as a function of the rescaled magnetic field $\Bad$ and setting $\epsilon\sub{F}=2000$ and $\Ef=0$. (b) Left: different contributions to the the Hall conductance $G_{xy}$ for Neumann boundary conditions. Right: dispersion diagrams associated with the points $A,B,C$ of the left panel, where the shaded regions correspond to energies below  $E\sub{F}$.}
	\label{fig-GR}
\end{figure}


\subsection{Boundary effects for the Hall conductivity}\label{sec-boundaryeff}
In the left panel of Figure~\ref{fig-GR}~(a) we plotted the Hall conductance $G_{xy}$ as a function of the magnetic field and for different boundary conditions. For completeness, in the right panel we also show the same plot for the Hall resistance $R_{xy}\equiv 1/G_{xy}$. Note that in each case we set $\Ef=0$. In order to perform the numerical calculation, we truncated the sum appearing in the expression~\eref{eq-Gxy} to the first five energy bands: as highlighted in the figure, this unavoidable approximation only affects the low-field behavior of $G_{xy}$, the latter being in turn associated with the classical regime. 

As expected, for sufficiently high values of the magnetic field the conductance decreases, as $B$ increases, in quantized steps: this behavior appears to be nearly independent of the boundary conditions, since it is related to the bulk spectrum. However, as we switch from Dirichlet to Neumann or Robin conditions, a finer structure arise in each transition region between two plateaux. As a matter of fact, this transition region can indeed be controlled by suitably tuning the boundary conditions. In particular, in the case of Neumann boundary conditions the different spectral contributions to the conductance (namely $\Delta G_{xy,\textup{bulk}}$, $\Delta G_{xy,\textup{bumps}}$ and $\Delta G_{xy,\textup{edges}}$) are emphasized in Figure~\ref{fig-GR}~(b).

The quantization of $G_{xy}$, \ie its bulk structure, is particularly manifest in the phase diagram of Figure~\ref{fig-GBE} (top row), where the conductance is plotted as a function of both the Fermi energy $E\sub{F}$ and the magnetic field $B$, but setting again $\Ef=0$. Conversely, in these plots the edge-structure of $G_{xy}$ is barely noticeable, being a finer effect of lower magnitude. In order to highlight this finer structure, in the bottom row of the same figure we have plotted another phase diagram associated with the  fractional part of $G_{xy}$, that is the quantity $G_{xy}-\lfloor G_{xy} \rfloor$, where $\lfloor x \rfloor$ denotes the floor of $x$, that is the largest integer less or equal to $x$. In this case the difference between Dirichlet and Neumann boundary conditions is much more evident, the latter showing wider transition regions.

\begin{figure}[tbp]
	\centering
	\includegraphics[width=0.9\textwidth]{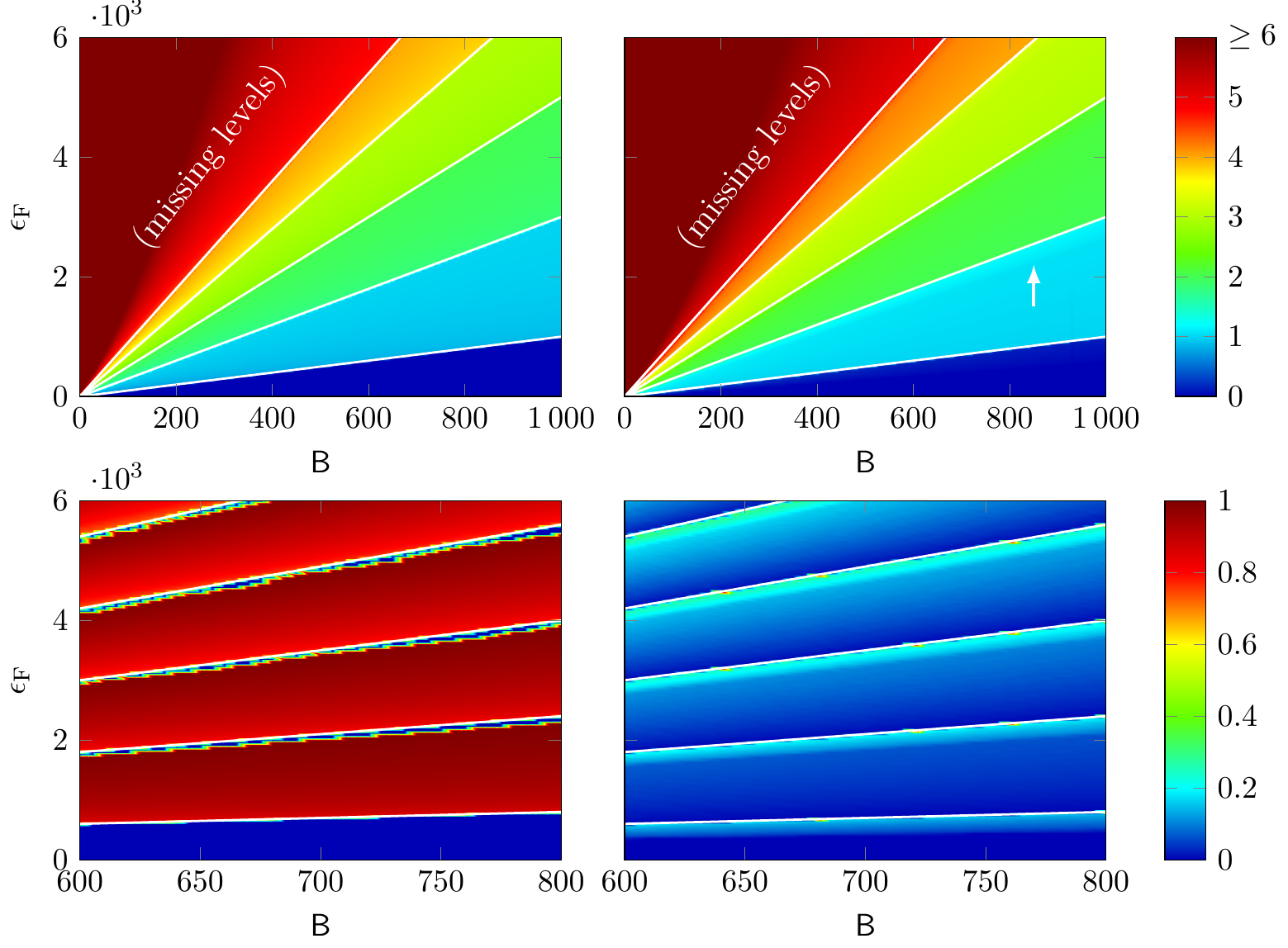}

\caption{Phase diagram of the Hall conductance $G_{xy}$ in units of $e^2/h$ (top row) and of its fractional part $G_{xy}-\lfloor G_{xy} \rfloor$ (bottom row) in the $\epsilon\sub{F}\Bad$-plane and for Dirichlet (left column) and Neumann (right column) boundary conditions; the arrow highlights the fine structure associated with Neumann boundary conditions, whereas the solid white lines represent the Landau levels $\epsilon_n^{\textup{HO}}=(2n+1)\Bad$.}

\label{fig-GBE}
\end{figure}

To conclude the numerical analysis, in Figure~\ref{fig-GWE} we plotted the Hall conductance as a function of the magnetic field, fixing Dirichlet boundary conditions but varying respectively the strip width $W$, with respect to a reference width given by $W_0\equiv [\hbar^2 \epsilon_{\text{F}}/(2mE_{\text{F}})]^{\frac{1}{2}}$, and the electric field $\Ef$. 
The left panel of the figure shows how the finite (transversal) size of the Hall system affects the quantization of the conductance: as the width of the strip increases, the slope of each nearly-flat plateau decreases, ultimately becoming flat in the limit $W\to\infty$ when the strip expands to the whole plane. On the other side, remarkably, the quantization of $G_{xy}$ is almost lost when $W/W_0$ is sufficiently small. The right panel of the figure shows another interesting phenomenon, that is the breakdown of the QHE at large values of the electric field: when $\Ef$ increases, indeed, the width of each plateau decreases, ultimately recovering the classical Hall regime in which $G_{xy}$ is inversely proportional to $B$.

\begin{figure}[tbp]
\centering
\includegraphics{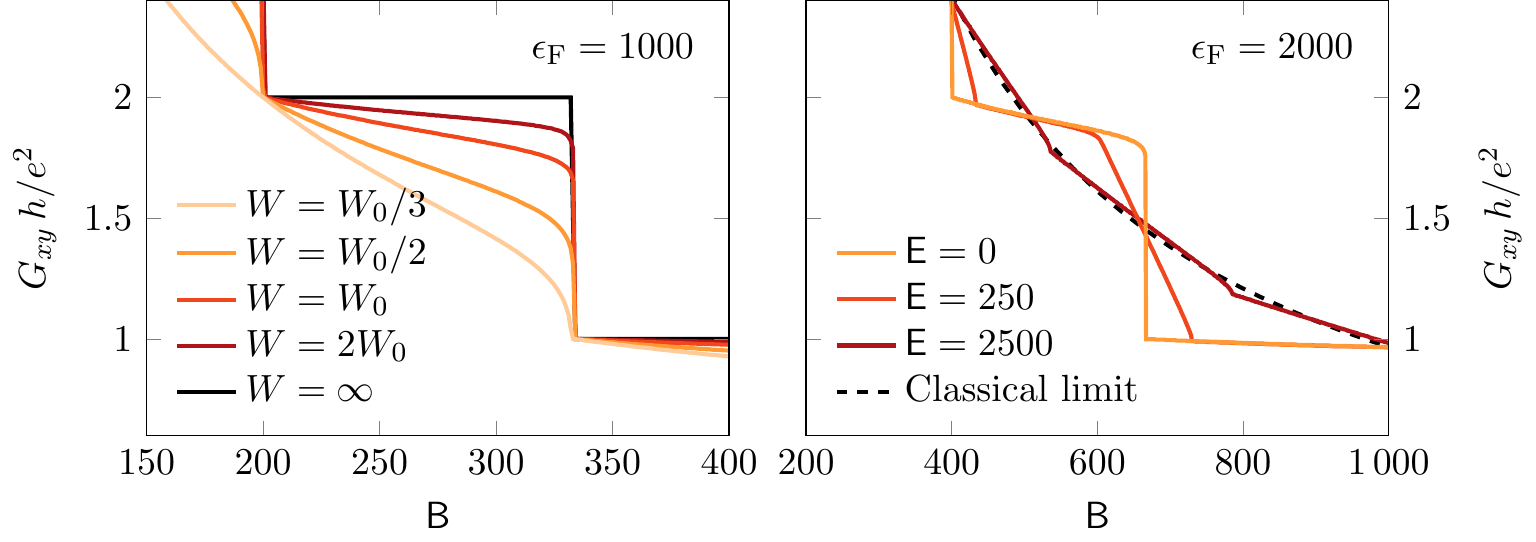}

\caption{Finite-size (left) and electric (right) effects for the Hall conductance $G_{xy}$; the reference width in the left figure is $W_0=[\hbar^2 \epsilon_{\text{F}}/(2mE_{\text{F}})]^{\frac{1}{2}}$, whereas the plane limit ($W=\infty$) has been adapted from~\cite{Kr04}.
}
\label{fig-GWE}
\end{figure}

\section{Conclusions}
In this work we formulated a self-consistent model of the integer QHE, using boundary conditions to investigate finite-size effects associated with the Hall conductivity. By assuming the invariance with respect to longitudinal translations, we were able to characterize the general spectral properties of the system. Then we focused  on the case of (fibered) Robin boundary conditions, which have been identified for physical reasons, \ie locality and translational invariance, and turn out to keep the problem tractable. By determining the spectrum (and the related velocity eigenvalues) corresponding to the selected boundary conditions, we have been able to predict a new kind of states with no classical analogues. The latter have a finite propagation velocity as classical edge states, but their chirality is the same of classical bulk states. Moreover, boundary conditions turned out to add a finer structure to the quantization pattern of the Hall conductivity: the integer plateau are substantially preserved, but a transition region between two consecutive plateaux appears, smoothly controlled by the boundary conditions. Finally, since we derived a formula for the Hall conductivity which depends exactly on the applied electric field, we also predicted the breakdown of the QHE. 

This work could be extended in different directions, which we may consider in the future. The first one regards the weakening of our assumptions leading to the fibered Robin conditions, and in particular the assumption of $k$-independent boundary conditions. In this more general setting, one should \eg consider a modified Hellmann-Feynman formula which keeps trace of the $k$-depending domain~\cite{EFC}. Another direction involves instead a more detailed physical characterization of the Robin boundary conditions~\cite{BeWa10}, which in turn may lead to their engineering in a real Hall device and to practical applications of this work. Besides, our analysis can potentially be extended to relativistic models related to the Dirac equation, which are usually employed in the description of graphene and related materials~\cite{Aso19}.

\section*{Acknowledgments}
We thank Giovanni Gramegna for carefully reading and revising an early version of the manuscript. 
G.A. thanks the Department of Theoretical Physics of the University of Zaragoza for its hospitality during the preparation of this work. 
This work was partially supported by Istituto Nazionale di Fisica Nucleare (INFN) through the project “QUANTUM” and the Italian National Group of Mathematical Physics (GNFM-INdAM). P.F. and D.L. acknowledge support by MIUR via PRIN 2017 (Progetto di Ricerca di Interesse Nazionale), project QUSHIP (2017SRNBRK). The work of M.A. and Y.M. is partially supported by Spanish MINECO/FEDER grant PGC2018-095328-B-I00 and DGA-FSE
grant 2020-E21-17R.

\appendix

\section{The Weber differential equation}\label{sec-weber}
In this appendix we report some standard choices of pair of independent solutions for the Weber differential equation and we explicitly show its relation with the fiber operator $h(k)$.

\subsection{Independent pairs of solutions}
The Weber equation is a second order differential equation given as follows:
\begin{equation}\label{eq-wde}
y''(x)-\Biggl(a+\frac{x^2}{4}\Biggl)y(x)=0\,,
\end{equation}
$x$ being a real variable; its solutions are known as \emph{parabolic cylinder} functions or \emph{Weber-Hermite} functions~\cite{Ba53}. An alternative expression which often appears in the literature is immediately recovered by setting $\nu=-a-\case{1}{2}$:
\begin{equation}\label{eq-wde2}
y''(x)+\Biggl(\nu+\frac{1}{2}-\frac{x^2}{4}\Biggl)y(x)=0\,.
\end{equation}
A simple pair of linearly independent solutions of Eq.~\eref{eq-wde} having definite parity (respectively even and odd) is the following:
	\begin{eqnarray}
	y_1(a,x)&=\e^{-\frac{1}{4}x^2}\,{}_1 F_1\Biggl(\frac{a}{2}+\frac{1}{4}, \frac{1}{2}; \frac{x^2}{2}\Biggl)\qquad&\textnormal{(even)} \label{eq-webersol1}\\
	y_2(a,x)&=x\e^{-\frac{1}{4}x^2}\,{}_1 F_1\Biggl(\frac{a}{2}+\frac{3}{4}, \frac{3}{2}; \frac{x^2}{2}\Biggl)&\textnormal{(odd)}\label{eq-webersol2}
	\end{eqnarray}
where ${}_1F_1(a,c; x)$ is the \emph{confluent hypergeometric} function, defined as
\begin{equation}
{}_1F_1(a,c; x)\equiv\sum_{n=0}^{+\infty} \frac{\Upgamma(a+n)}{\Upgamma(a)}\frac{\Upgamma(c)}{\Upgamma(c+n)}\frac{x^n}{n!}\,,
\end{equation}
$\Upgamma(x)$ being Euler's Gamma function. However, the solutions~\eref{eq-webersol1}--\eref{eq-webersol2} display similar asymptotic behaviors at infinity; this phenomenon may complicate numerical computations for large arguments. Another independent pair of solutions is
	\begin{eqnarray}
	U(a, x)&=\frac{1}{ \sqrt{2^{b} \pi} }\biggl[\cos(b\pi)\Upgamma(\case{1}{2}-b) y_1 -\sqrt{2}\sin(b\pi)\Upgamma(1-b)y_2 
	\biggr]\\
	V(a, x)&=\frac{1}{\sqrt{2^{b} \pi}}
	\Biggl[\sin(b\pi)\frac{\Upgamma(\case{1}{2}-b)}{\Upgamma(\case{1}{2}-a)} y_1 +\sqrt{2}\cos(b\pi)\frac{\Upgamma(1-b)}{\Upgamma(\case{1}{2}-a)}y_2 
	\Biggl]
	\end{eqnarray}
where for compactness we have set $b=\case{a}{2}+\case{1}{4}$, see~\cite{AS64}; their asymptotic behavior for $x\gg|a|$ is given by
\begin{equation}
U(a,x)\sim\e^{-\frac{1}{4}x^2}x^{-a-\frac{1}{2}}\qquad\textnormal{and}\qquad 
V(a,x)\sim \bigl(\case{2}{\pi}\bigr)^{\frac{1}{2}}\e^{\frac{1}{4}x^2}x^{a-\frac{1}{2}}\,,
\end{equation}
thus avoiding the aforementioned numerical problem. Note that, as long as $a$ is real, all the above solutions are real; in particular if $\nu=-a-\case{1}{2}$ is a positive integer $n$ then $U(a, x)$ and $V(a, x)$ can be expressed in terms of the \emph{Hermite polynomials} $\Herm_n(x)$:
\begin{eqnarray}
U(-n-\case{1}{2}, x)&=\frac{\e^{-\frac{1}{4}}x^2}{\sqrt{2^n}} \Herm_n(x/\!\sqrt{2})\,,\\
V(n+\case{1}{2}, x)&=\frac{\e^{\frac{1}{4}}x^2}{\sqrt{2^n}} (-\iu)^n \Herm_n(\iu x/\!\sqrt{2})\,.
\end{eqnarray}
The Weber equation in the alternative form~\eref{eq-wde2} admits 
\begin{equation}
D_\nu(x)=U\bigl(-\nu-\case{1}{2}, x\bigr)\,,
\end{equation}
the so-called \emph{Whittaker} function, as a solution. There are two main possibilities to construct a pair of independent solutions starting from $D_\nu(x)$~\cite{Ba53}. The first pair is $D_\nu(x)$ and $D_{-\nu-1}(\iu x)$: they are always independent, but since $D_\nu(x)$ is real, $D_{-\nu-1}(\iu x)$ in general is complex. The other pair is given by $D_\nu(x)$ and $D_\nu(-x)$: these solutions are always real, but are linearly independent only if $\nu$ is not an integer.

\subsection{Weber equation from the eigenvalue equation}
As we have shown in Section~\ref{sec-model}, in the Landau gauge the Hall Hamiltonian $H$ defined on the strip $\Os$ can be decomposed as a direct integral, each fiber operator $h(k)$ representing an effective one-dimensional Hamiltonian. Since under our assumptions the spectrum of $h(k; B, \Ef)$ can be readily obtained from that of $h(k; B, 0)$ by using Eq.~\eref{eq-EnElectic}, we only need to solve the eigenvalue equation~\eref{eq-eigeneq} in the non-electric case, namely:
\begin{equation}\label{eq-eigprob1}
\bigl[h(k; B, 0)-E(k; B, 0)\bigr]\psi(k; y)=0\,.
\end{equation}
Multiplying by $-2mW^2/\hbar^2$ and scaling $y$ to $\yy\equiv y/W$, we obtain:
\begin{equation}\label{eq-eigprob2}
\Biggl[\frac{\dd^2}{\dd \yy^2}-\Bad^2\Biggl(\yy-\frac{\kad}{\Bad}\Biggr)^2+\epsilon\Biggr]\varphi(\yy)=0\,,
\end{equation}
where we have introduced the dimensionless parameters
\begin{equation}
\Bad\equiv\frac{eBW^2}{\hbar}\,,\qquad k\equiv kW\qquad\textnormal{and}\qquad \epsilon\equiv\frac{2mEW^2}{\hbar^2}\,, 
\end{equation}
and we have set $\varphi(\yy)\equiv\psi(W\yy)$ (omitting the dependence on $k$). Notice that the rescaled wavefunction $\varphi(\xi)$ is now confined in the unitary interval $(-1/2, 1/2)$. As a last step, we further divide Eq.~\eref{eq-eigprob2} by $2\Bad$ and rescale $\yy$ to $\tilde\yy\equiv \sqrt{2\Bad}(\yy-\kad/\Bad)$, thus obtaining
\begin{equation}
\Biggl[\frac{\dd^2}{\dd \tilde\yy^2}-\frac{{\tilde\yy}^2}{4}+\frac{\epsilon}{2\Bad}\Biggr]\tilde\varphi(\tilde\yy)=0\,,\qquad\tilde\varphi(\tilde\yy)\equiv\varphi\Biggl(\frac{\tilde\yy}{\sqrt{2\Bad}}+\frac{\kad}{\Bad}\Biggr)\,,
\end{equation}
which is the Weber equation~\eref{eq-wde} with parameter $a=-\epsilon/(2\Bad)$. At this point, denoting with $u^1(a, x)$ and $u^2(a, x)$ a suitable pair of independent solutions of the Weber equation, the general solution is given by
\begin{equation}
\tilde\varphi(\tilde\yy)=c_1 u^1(a, \tilde\yy)+c_2 u^2(a, \tilde\yy)\,,
\end{equation}
$c_1$ and $c_2$ being two complex constants. The original non-rescaled wavefunction $\psi(k; y)$ which solves Eq.~\eref{eq-eigprob1} accordingly reads
\begin{eqnarray}
\psi(k; y)&=\sum_{i=1}^{2} c_i u^i\Biggl(-\frac{\epsilon}{2\Bad}, \sqrt{2\Bad}\Biggl(\frac{y}{W}-\frac{\kad}{\Bad}\Biggr)\Biggr)\nonumber\\
&=\sum_{i=1}^{2} c_i u^i\Biggl(-\frac{E}{\hbar\omega_B}, \sqrt{2} \frac{y-kl_B^2}{l_B}\Biggr)\,.
\end{eqnarray}
Note that the above solution is actually independent of $W$: such dependence is indeed introduced only when one imposes the boundary conditions.

\section{Regularity of the fiber resolvent}\label{sec-resolvent}
In this appendix we prove a regularity result involving the \emph{resolvent} of $h_{U}(k)$. In particular, given a family $\{U(k)\}_{k\in\mathbb{R}}$ of unitary matrices such that the function
\begin{equation}
		k\in\mathbb{R}\mapsto U(k)\in\UU(2)
\end{equation}
is measurable, we will prove that the function
\begin{equation}
k\in\R \mapsto (h_{U}(k)-z)^{-1}\in\mathcal{B}\left(\Leb^2(I_W)\right)
\end{equation}
is strongly measurable (and thus also weakly measurable) for all $z\in\mathbb{C}\setminus\mathbb{R}$.

We will achieve this result by using the following \emph{Krein formula}~\cite{AlbPan05}
\begin{equation}\label{eq-krein}
(h_{U}(k)-z)^{-1}=(h_{I}(k)-z\bigr)^{-1}-K(k; z)\,,
\end{equation}
where $K(k; z)$ is a bounded operator to be defined in the following; this equation, which holds for each $k\in\R$ and $z\in\C\setminus\R$, relates the resolvent of the (arbitrary) self-adjoint extension $h_{U}(k)$ with the resolvent of the Dirichlet extension $h_{I}(k)$.
\subsection{Krein formula}
For each $\psi\in H^2(I_W)$, let us define the trace operators
\begin{equation}
\Gamma_1\psi\equiv
\left(\begin{array}{@{}c@{}}
\psi(-W/2)\\ \psi(W/2)
\end{array}\right)
\qquad\textnormal{and}\qquad \Gamma_2\psi \equiv l_0
\left(\begin{array}{@{}c@{}}
-\psi'(-W/2)\\ \psi'(W/2)
\end{array}\right)\,,
\end{equation}
respectively corresponding to the quantities $\Psi$ and $\Psi'$ of Eq.~\eref{eq-boundarydata}, and the related operators
\begin{equation}
\gamma(k;z)\colon \C^2\to  H^2(I_W)\,,\qquad \gamma(k;z)\equiv \left(\Gamma_1|_{\ker(h^{\dagger}(k)-z)}\right)^{-1}
\end{equation}
and
\begin{equation}
Q(k; z)\colon \C^2\to \C^2\,,\qquad Q(k; z)\equiv \Gamma_2\gamma(k;z)\,.
\end{equation}
Then, given a matrix $U(k)\in\UU(2)$ for each $k\in\R$, we also define the related matrices 
\begin{equation}
A(k)\equiv \iu (I+U(k))\qquad\textnormal{and}\qquad B(k)\equiv I-U(k)\,.
\end{equation}
At this point, the operator $K(k;z)$ of Eq.~\eref{eq-krein} can be expressed as
\begin{equation}
K(k; z)\equiv \gamma(k;z)\bigl(B(k)Q(k; z)-A(k)\bigr)^{-1}B(k)\gamma^{\dagger}(z^*)\,,
\end{equation}
see Eq.~(8) of~\cite{AlbPan05}.

\subsection{Resolvent regularity}
As follows from the Krein formula, the regularity of $k\mapsto (h_{U}(k)-z)^{-1}$ depends on the regularity of both 
\begin{equation}
k\mapsto (h_{I}(k)-z)^{-1}\qquad\textnormal{and}\qquad k\mapsto K(k; z)\,.
\end{equation}
For what concerns the Dirichlet resolvent, accordingly to Eq.~\eref{eq-bc-kdep} the operator $h_{I}(k)$ is defined, for each $k\in \R$, on the common domain (\emph{core})
\begin{equation}
\mathfrak{D}_I \equiv \{\psi\in  H^2(I_W) : \Psi= 0\}\,.
\end{equation}
Therefore, since the potential $V_k(y)$ appearing in the expression of $h(k)$ is a continuous function of $k$, the strong continuity (and hence measurability) of the resolvent $(h_{I}(k)-z)^{-1}$ is a straightforward corollary of Lemma~6.36 of~\cite{Teschl}. Remarkably, this result does not depend on the regularity of the function $k \mapsto U(k) $.

Moving to the operator $K(k;z)$, we preliminary observe that
\begin{equation}
\psi\in \ker(h^{\dagger}(k)-z)\Longleftrightarrow \psi(y)=c_1 u^1_{z}(k; y) + c_2 u^2_{z}(k; y)\,,
\end{equation}
where  $c_1,c_2\in\C$ and where $u^1_{z}(k;y)$ and $u^2_{z}(k;y)$, being two suitable independent solutions of the Weber equation, are regular functions of $k$ (see the previous appendix). Therefore, by computing the action of the operators $\gamma(k;z)$ and $Q(k;z)$, we find that they are regular functions of $k$, again independently of the function $k\mapsto U(k)$.

The only terms remaining are thus the matrices $A(k)$ and $B(k)$, which however trivially inherit the regularity of $k\mapsto U(k)$; this completes the claim.


\section{Translational symmetry and fibered boundary conditions}\label{sec-symmetry}
In this appendix we show that fibered boundary conditions for the Hall Hamiltonian $H$ emerge as a natural consequence of a symmetry request: indeed, a realization of $H$ turns out to be decomposable as a direct integral of self-adjoint fibers if and only if it is invariant under longitudinal translations, in a sense that will be shortly clarified.

\subsection{Commuting self-adjoint operators}

Let us consider two (possibly) unbounded operators $A,B$ on a Hilbert space $\mathcal{H}$. Since we are dealing with unbounded operators, the equation
\begin{equation}
	[A,B]=0
\end{equation}
is generally \textit{ill-defined}. Luckily, for self-adjoint operators, a proper notion of commutativity with all desired implications does exist (see Chapter~VIII.5 of~\cite{RS1}): two self-adjoint unbounded operators $A$, $B$ are said to commute whenever their \emph{PVMs} commute, and thus, for all bounded \emph{Borel} functions $f,g\colon \mathbb{R}\rightarrow\mathbb{C}$, we have
\begin{equation}
\bigl[f(A),g(B)\bigr]=0\,.
\end{equation}
Let us recall here a useful property:
\begin{prop}[\cite{RS1}, Theorem VIII.13]\label{prop:comm0}
	Let $A, B$ be two self-adjoint unbounded operators. The following conditions are equivalent:
	\begin{enumerate}
		\item $A$ and $B$ commute;
		\item for all $s,t\in\mathbb{R}$,
		\begin{equation}\label{eq:comm}
		\bigl[\e^{-\iu sA},\e^{-\iu tB}\bigr]=0\,;
		\end{equation}
		\item for all $z,w\in\mathbb{C}$ with $\mathrm{Im}\, z,\mathrm{Im}\, w\neq0$,
		\begin{equation}\label{eq:comm2}
		\bigl[(A-z)^{-1},(B-w)^{-1}\bigr]=0\,.
		\end{equation}
	\end{enumerate}
\end{prop}
For our purposes, it will be useful to state a slightly different equivalence:
\begin{prop}\label{prop:comm}
	$A$ and $B$ commute if and only if, for all bounded Borel functions $f\colon \mathbb{R}\rightarrow\mathbb{C}$,
	\begin{equation}\label{eq:comm3}
	\bigl[f(A),(B+\iu)^{-1}\bigr]=0\,.
	\end{equation}
\end{prop}
\begin{proof}
	If $A$ and $B$ commute, all bounded functions of them commute and thus Eq.~\eref{eq:comm3} holds \textit{a fortiori}. Conversely, if Eq.~\eref{eq:comm3} holds, then we also have
	\begin{equation}\label{eq:comm4}
	\bigl[f(A),(B-w)^{-1}\bigr]=0
	\end{equation}
	for all $z$ with $\mathrm{Im}\, z\neq0$, as can be shown by expressing $(B-w)^{-1}$ as a power series in $w$ around $\iu$; since Eq.~\eref{eq:comm4} holds for every bounded Borel function $f$,  Eq.~\eref{eq:comm2} holds as a particular case, and so Proposition~\ref{prop:comm0} implies the claim.
\end{proof}
Obviously, the roles of $A$ and $B$ in the above Proposition are completely interchangeable.

\subsection{Commutation and longitudinal invariance}
Coming back to our original problem, let $H_{\mathfrak{D}}$ be a (not necessarily fibered) self-adjoint extension of the Hall Hamiltonian on some domain $\mathfrak{D}\subset L^2(\mathbb{R})\otimes L^2(I_W)$. Besides, let $p_x$ be the momentum operator on the real line,
\begin{equation}
	p_x=-\iu\hbar\frac{\mathrm{d}}{\mathrm{d}x}\,,
\end{equation}
which is itself a self-adjoint operator if we choose the first Sobolev space $H^1(\mathbb{R})$ as its domain. Recall that $p_x$ is the generator of translations on the real line, that is, for every $a\in\mathbb{R}$ and $\psi\in L^2(\mathbb{R})$, we have that
\begin{equation}
	\left(\e^{-\iu ap_x}\psi\right)(x)=\psi(x-a)\,.
\end{equation}
Consequently, $P_x\equiv p_x\otimes I$ represents the generator of longitudinal translations on the strip $\Os$: for every $a\in\mathbb{R}$ and $\psi\in L^2(\mathbb{R})\otimes L^2(I_W)$,
\begin{equation}
\left(\e^{-\iu aP_x}\psi\right)(x,y)=\psi(x-a,y)\,.
\end{equation}
We shall prove that $H_{\mathfrak{D}}$ and $P_x$ commute, in the sense previously discussed, if and only if $H_{\mathfrak{D}}$ can be written, up to a unitary transformation, as a \textit{decomposable} operator, \ie if it admits a direct integral representation:
\begin{equation}
	H_{\mathfrak{D}}\cong \int_{\mathbb{R}}^{\oplus}h(k)\dd{k}\,,
\end{equation}
with $\{h(k)\}_{k\in\R}$ a certain family of self-adjoint operators on $L^2(I_W)$.
First of all, as discussed in the main text, it will be convenient to switch from the $(x,y)$-representation to the $(k,y)$-representation; by performing a partial Fourier transform, the operator $P_x$ is mapped in the operator $\hat{P}_x$, which simply acts as the multiplication by $k$, namely
\begin{equation}\label{eq:hatp0}
\bigl(\hat{P}_x\psi\bigr)(k,y)=k\,\psi(k,y)\,,
\end{equation}
while $H_{\mathfrak{D}}$ is mapped in the operator $\hat{H}_{\mathfrak{D}}$ with domain $\mathcal{F}_x\mathfrak{D}$, on which it acts as
\begin{equation}
	(\hat{H}_{\mathfrak{D}}\psi)(k,y)=-\frac{\hbar^2}{2m}\frac{\partial^2}{\partial y^2}\psi(k,y)+V_k(y)\psi(k,y)\,,
\end{equation}
where $	V_k(y)\equiv \frac{1}{2}m\omega_B^2(y-kl^2_B)+e\mathcal{E}y$, see Eq.~\eref{eq-actiontildeHBE}. Both operators act on the Hilbert space $L^2(\hat{\mathbb{R}})\otimes L^2(I_W)$, which is naturally isomorphic to the Bochner space $L^2(\hat{\mathbb{R}}; L^2(I_W))$: the function
\begin{equation}
	(k,y)\in\mathbb{R}\times I_W\mapsto \psi(k,y)\in\mathbb{C}
\end{equation}
is indeed naturally associated with the $L^2(I_W)$-valued function
\begin{equation}
	k\in\mathbb{R}\mapsto\psi_k\in L^2(I_W)
\end{equation}
by defining $\psi_k(y)=\psi(k,y)$ for all $k\in\mathbb{R}$ and $y\in I_W$. With this identification, the operator $\hat{P}_x$ is thus decomposable as a direct integral:
\begin{equation}\label{eq:hatp}
	\hat{P}_x=\int_{\mathbb{R}}^\oplus k\dd{k}\,.
\end{equation}

The operator $\hat{H}_{\mathfrak{D}}$ can be interpreted, as well, as an operator on $L^2(\hat{\mathbb{R}}; L^2(I_W))$; in general, however, it is not guaranteed to be decomposable since its domain $\mathcal{F}_x\mathfrak{D}$ (\ie its boundary conditions) is not guaranteed to be compatible with the direct integral structure. We now show that, indeed, this holds if and only if the invariance under longitudinal translations holds.

\begin{prop}
	$H_{\mathfrak{D}}$ commutes with $P_x$ if and only if $\hat{H}_{\mathfrak{D}}$ is decomposable.
\end{prop}
\begin{proof}
	By Theorem XIII.85(b) of~\cite{RS4}, $\hat{H}_{\mathfrak{D}}$ is decomposable if and only if $(\hat{H}_{\mathfrak{D}}+\iu)^{-1}$ is decomposable; in turn, by Theorem XIII.84 of~\cite{RS4}, the latter is decomposable if and only if it commutes with all bounded decomposable operators on $L^2(\hat{\mathbb{R}};L^2(I_W))$ whose fibers are multiple of the identity, i.e. with all operators that can be written as
	\begin{equation}\label{eq:tf}
	\int_{\mathbb{R}}^\oplus f(k)\dd{k}
	\end{equation}
	for some bounded Borel function $f$; by construction, however, such an operator is nothing but the multiplication operator on $L^2(\hat{\mathbb{R}};L^2(I_W))$ associated with $f$, which means that the operator in \eref{eq:tf} is in fact $f(\hat{P}_x)$. Consequently, Theorem XIII.84 of~\cite{RS4} can be restated as follows: $\hat{H}_{\mathfrak{D}}$ is decomposable if and only if, for all bounded Borel functions $f$,
	\begin{equation}
		\bigl[f(\hat{P}_x),(\hat{H}_{\mathfrak{D}}+\iu)^{-1}\bigr]=0\,;
	\end{equation}
	this, by Proposition~\ref{prop:comm}, holds if and only if $\hat{H}_{\mathfrak{D}}$ and $\hat{P}_x$ commute, and thus if and only if $H_{\mathfrak{D}}$ and $P_x$ commute. 
\end{proof}
\subsection{Floquet theorem and Bloch sectors}\label{sec:floquet}
In real systems the effects of the metal structure can be encoded by a perturbation of the Hamiltonian \eref{eq-HAE} which in the case of perfect crystalline structure is given by a periodic potential $V$. If we neglect the effect of this perturbation in the transverse direction the new Hamiltonian reads
\begin{equation}\label{eq-HAEV}
H_{\A, V}\equiv -\frac{\hbar^2}{2m}\gradA^{2}+e\Ef y\,+ V(x),
\end{equation}
where $V$ is a one-dimensional periodic function
\begin{equation}\label{eq-potential}
 V(x-a)=V(x),
\end{equation}
$a$ being the distance between crystal nodes.
The residual translation invariance of the Hamiltonian \eref{eq-HAEV}
implies that 
\begin{equation}\label{eq-dtrans}
[H_{\A, V}, T_a]=0,
\end{equation}
where $T_a$ is the operator 
\begin{equation}\label{eq-trans}
 T_a\psi(x)= \psi(x-a).
\end{equation}
that generates discrete translations. The set of these translations ${\cal{T}}_a=\{T_a^n, n\in \Z\}$ is thus an abelian symmetry group of the Hamiltonian \eref{eq-potential}.

Since ${\cal{T}}_a$ is abelian its irreducible unitary representations are one-dimensional and can be parametrized by a phase factor
\begin{equation}\label{eq-fase}
 T_a\psi(x)= e^{i\alpha}\psi(x)=e^{i\kappa a} \psi(x)\,,
\end{equation}
where $\kappa\in{[-\pi/ a, \pi/ a]}=B_a$,  $B_a$ being the Brillouin domain
of the Bloch phase $\kappa$.

Thus,  all energy levels of the Hamiltonian \eref{eq-HAEV} can be decomposed into irreducible representations of ${\cal{T}}_a$ satisfying the pseudoperiodic boundary conditions \eref{eq-fase}.

\begin{prop}
The Floquet theorem establishes that any self-adjoint extension of the Hall Hamiltonian \eref{eq-HAEV} defined on some domain $\mathfrak{D}\subset L^2(\mathbb{R})\otimes L^2(I_W)$ can be decomposed as a fibered sum of Hamiltonians $h_{\mathfrak{D},a}$
\begin{equation}
	H^{\mathfrak{D}}_{\A, V}\cong\int_{B_a}^{\oplus}h_{\mathfrak{D},a}(\kappa)\dd{\kappa}\,,
\end{equation}
defined on $L^2({\mathbb{S}}^1)\otimes L^2(I_W)$ by the boundary conditions of $H^{\mathfrak{D}}_{\A, V}$ and \eref{eq-fase}.
\end{prop}
The fibered sum of Hamiltonians is associated with  energy bands. Usually, the spectral structure is given by these continuous bands separated by spectral gaps which give rise to the insulator regime. However, in the special case $V=0$ considered in the previous appendix the period $a$ is arbitrary and there is no gaps between the Bloch bands associated with the Floquet decomposition, i.e. the spectrum of  $	H^{\mathfrak{D}}_{\A, V}$ contains the half-line $[E_0,\infty)$, with the value $E_0$ depending on the boundary conditions.

\section{Quantum Hall conductivity}\label{sec-conductivity}

In this appendix we give  an alternative derivation of the formula for the conductance~\eref{eq-Gxy}
which does not make use of Drude theory.

When the electric field is given by $\Eb=(0,\Ef,0)$, the Hall conductivity is defined by the ratio $\sigma_{xy}=\partial J/\partial\Ef$, $J=-env_{x}$ being the longitudinal current density. The current density of a quantum system can be computed by summing the velocity of each state below the Fermi energy $E\sub{F}$, that is,
\begin{eqnarray}\nonumber
\hskip-30pt
J(E\sub{F})&\equiv -\frac{e}{V} \sum_{n=0}^{+\infty} \int_{-\infty}^{E\sub{F}-\hbar k\frac{\Ef}{B}-\frac{m}{2}\frac{\Ef^2}{B^2}}\dd{E}\int_{-\infty}^{+\infty} \ddelta(E-E_n(k))\,v_n(k)\frac{L}{2\pg}\dd{k}\\
&=-\frac{e}{V} \sum_{n=0}^{+\infty} \int_{-\infty}^{+\infty}  \htheta\Biggl(E\sub{F}-\hbar k\frac{\Ef}{B}-\frac{m}{2}\frac{\Ef^2}{B^2}-E_n(k)\Biggr)\,v_n(k)\frac{L}{2\pg}\dd{k}\nonumber \,,
\end{eqnarray}
where $v_n(k)=\dd{E_n}/(\hbar \dd{k})$ is the $n$-th band velocity defined in Eq.~\eref{eq-vnk}. Note that this current density depends on $\Ef$  through both $E_n(k)$ and $v_n(k)$. Thus,
\begin{eqnarray}\nonumber
\hskip-30pt
J(E\sub{F})&=-\frac{e}{V} \sum_{n=0}^{+\infty} \int_{-\infty}^{+\infty}  \htheta\Biggl(E\sub{F}-\hbar k\frac{\Ef}{B}-\frac{m}{2}\frac{\Ef^2}{B^2}-E_n(k)\Biggr)\,\frac{\partial{E_n}}{\partial{k}}\frac{L}{2 \pg\hbar}\dd{k}\,,\\
&=-\frac{e L}{h V } \sum_{n=0}^{+\infty} (E_n(k_{n, +})- E_n(k_{n, -})) =\frac{e L}{2\pi V B } \sum_{n=0}^{+\infty}  {\cal{E}}\Delta k_n
\end{eqnarray}
where $k_{n,\pm}$ are the solutions of the equation
\begin{equation}
E_n(k_{n, \pm})= E\sub{F}-\hbar\frac{\Ef}{B}  k_{n, \pm}-\frac{m}{2}\frac{\Ef^2}{B^2}
\end{equation}
and
\begin{equation}
\Delta k_n=k_{n, +}-k_{n, -}.\end{equation}

The Hall conductivity is then given by
\begin{equation}
\sigma_{xy}=\frac{\partial J} {\partial\Ef}=\frac{e L}{2\pi V B} \sum_{n=0}^{+\infty}  {}\Delta k_n + {\cal O}({\cal E})
\end{equation}
and the Hall conductance by 
\begin{equation}
G_{xy}=\frac {V}{W L} \sigma_{xy}=\frac{e}{2\pi W B} \sum_{n=0}^{+\infty}  {}\Delta k_n  + {\cal O}({\cal E}).
\end{equation}
\section{Fibered boundary conditions in the position representation}\label{sec-position}
In this appendix we sketch how fibered Robin boundary conditions with a $k$-independent Robin parameter $\alpha$ actually correspond to \emph{global} Robin boundary conditions (in the position representation), also discussing the less trivial case in which $\alpha=\alpha(k)$ depends linearly on $k$.

For a suitably regular function $\psi(x,y)$ living on the strip $\Os=\R\times I_W$, let us introduce the \emph{chiral boundary condition}
\begin{equation}\label{eq-chiral}
\pm \frac{\partial\psi}{\partial y}(x,y)\biggl|_{y=\pm W/2}=\alpha_0 \psi(x,\pm W/2)-\iu\alpha_1\frac{\partial\psi}{\partial x}(x,\pm W/2)\,,
\end{equation}
where $\alpha_0,\alpha_1\in\R$ (see \eg~\cite{quanBill} for a discussion of its physical significance); note in particular that chiral conditions reduce to Robin conditions if $\alpha_1=0$. The above expression is manifestly invariant under longitudinal translations, which are indeed generated by the longitudinal momentum operator $p_x=-\iu\hbar \partial/\partial_x$. Therefore, we expect chiral boundary conditions to be associated with some fibered boundary conditions, in the mixed $(k,y)$-representation. To show this, let us take
\begin{equation}
\psi(x,y)=(\mathcal{F}_x^{-1}\phi)(x,y)=\frac{1}{\sqrt{2\pg}}\int_{\R} \e^{\iu kx}\phi(k,y)\dd{k}\,.
\end{equation}
By plugging this in Eq.~\eref{eq-chiral}, we get the following boundary condition for $\phi(k,y)$:
\begin{equation}
\pm \frac{\partial \phi}{\partial y}(k,y)\biggl|_{y=\pm W/2}=(\alpha_0+\alpha_1 k)\phi(k,\pm W/2)\,.
\end{equation}
This boundary condition can thus be understood as a fibered Robin condition with an effective Robin parameter $\alpha(k)=\alpha_0+\alpha_1k$ depending linearly on $k$; in particular, it reduces to the $k$-independent fibered Robin condition \eref{eq-Robin} when $\alpha_1=0$, as expected.

\end{document}